\documentclass[twocolumn]{autart}    

\usepackage{lbase}
\usepackage{graphicx}          
\usepackage{blindtext}
\usepackage{amssymb}
\usepackage{mathtools}
\usepackage{blox}
\usepackage{eurosym}
\usepackage{bm}
\usepackage{algorithm}
\usepackage[noend]{algorithmic}
\usepackage{upgreek}
\usepackage{color}
\usepackage{siunitx}	
\usepackage{booktabs}
\usepackage{tabularx}						
\usepackage{longtable}	
\usepackage{multirow}	
\usepackage{cases}	
\usepackage{balance}
\usepackage{accents} 
\usepackage{xcolor}
\usepackage[normalem]{ulem}
\usepackage{cite}
\usepackage{type1ec}
\usepackage{mathptmx}
\usepackage{bbold}
\usepackage[hidelinks]{hyperref}
\usepackage{xurl}
\DeclareMathAlphabet{\mathcal}{OMS}{cmsy}{m}{n}

\newcommand{\R}{\mathbb{R}}
\newcommand{\mP}{\mathcal{P}}
\newcommand{\mG}{\mathcal{G}}

\newcommand{\mRid}{\mathcal{R}_0}
\newcommand{\Rid}{R_{0}}
\newcommand{\Arid}{A_{0}}
\newcommand{\Brid}{B_0}
\newcommand{\Crid}{C_0}
\newcommand{\Drid}{D_0}

\newcommand{\uid}{u_0}
\newcommand{\mRdd}{\hat{\mathcal{R}}_0}
\newcommand{\Rdd}{\hat{R}_0}
\newcommand{\Ardd}{\hat{A}_0}
\newcommand{\Brdd}{\hat{B}_0}
\newcommand{\Crdd}{\hat{C}_0}
\newcommand{\Drdd}{\hat{D}_0}
\newcommand{\udd}{\hat{u}_0}
\newcommand{\xrdd}{\hat{z}}
\newcommand{\ndd}{\hat{n}_0}

\newcommand{\Br}{B_{\text{r}}}

\newcommand{\ueq}{u_{1\text{eq}}}

\newcommand{\M}{\mathcal{M}}
\newcommand{\xmd}{\dot{x}_\text{m}}
\newcommand{\xm}{x_\text{m}}
\newcommand{\Am}{A_{\text{m}}}
\newcommand{\Bm}{B_{\text{m}}}
\newcommand{\Cm}{C_{\text{m}}}
\newcommand{\xo}{x_{\text{o}}}
\newcommand{\yo}{y_{\text{o}}}
\newcommand{\yod}{\dot{y}_{\text{o}}}
\newcommand{\eo}{e_{\text{o}}}
\newcommand{\nm}{n_{\text{m}}}
\newcommand{\mMdd}{\hat{\mathcal{M}}}
\newcommand{\Amdd}{\hat{A}_{\text{m}}}
\newcommand{\Bmdd}{\hat{B}_{\text{m}}}
\newcommand{\Cmdd}{\hat{C}_{\text{m}}}
\newcommand{\Budd}{\hat{B}_{\text{u}}}
\newcommand{\xmddd}{\dot{\hat{x}}_{\text{m}}}
\newcommand{\xmdd}{\hat{x}_{\text{m}}}
\newcommand{\xdd}{\hat{x}_\text{o}}
\newcommand{\xddd}{\dot{\hat{x}}_\text{o}}
\newcommand{\ymdd}{\hat{y}_\text{o}}
\newcommand{\ymddd}{\dot{\hat{y}}_\text{o}}
\newcommand{\edd}{\hat{e}_\text{o}}
\newcommand{\Mdd}{\hat{M}(s)}

\newcommand{\urm}{\Check{u}_{\text{o}}}
\newcommand{\sigmaddo}{\hat{\sigma}_{\text{0}}}
\newcommand{\sigmadd}{\hat{\sigma}}
\newcommand{\zetadd}{\hat{\zeta}}
\newcommand{\xa}{x_{\mathrm{a}}}
\newcommand{\xad}{\dot{x}_{\mathrm{a}}}
\newcommand{\Aa}{A_{\mathrm{a}}}
\newcommand{\BaU}{B_{\mathrm{a}}}
\newcommand{\Ca}{C_{\mathrm{a}}}
\newcommand{\BaR}{B_{\mathrm{r}}}

\newcommand{\mL}{\mathcal{L}}
\newcommand{\vecto}{\text{vec}}
\newcommand{\Tr}{\text{tr}}
\newcommand{\jO}{\mathrm{j}\omega}

\let\Theta\varTheta
\newcommand{\tg}{\forall t \geq 0}
\makeatletter
\xdef\@endgadget#1{{\unskip\nobreak\hfil\penalty50\hskip1em\hbox{}\nobreak\hfil#1\parfillskip=0pt\finalhyphendemerits=0\par}}
\newcommand\@Endofsymbol{$\triangledown$}
\newcommand\Endofremark{\@endgadget{\@Endofsymbol}}
\makeatother

\usepackage[many]{tcolorbox}
\usetikzlibrary{calc}
\tcbuselibrary{skins}
\newtheorem{problem}{Problem}
\newtcolorbox{resp}[1][]{%
	enhanced jigsaw,%
	colback=gray!2!white,%
	colframe=gray!90!black,%
	size=small,%
	boxrule=1pt,%
	halign title=flush center,%
	coltitle=black,%
	breakable,%
	drop shadow=black!70!white,%
	attach boxed title to top left={xshift=1cm,yshift=-\tcboxedtitleheight/2,yshifttext=-\tcboxedtitleheight/2},%
	minipage boxed title=3cm,%
	boxed title style={%
		colback=white,%
		size=fbox,%
		boxrule=1pt,%
		boxsep=2pt,%
		underlay={%
			\coordinate (dotA) at ($(interior.west) + (-0.5pt,0)$);
			\coordinate (dotB) at ($(interior.east) + (0.5pt,0)$);
			\begin{scope}[gray!80!black]
				\fill (dotA) circle (2pt);
				\fill (dotB) circle (2pt);
			\end{scope}
		}%
	},%
	#1%
}

\newtheorem{theorem}{Theorem}
\newtheorem{lemma}[theorem]{Lemma}

\newtheorem{proposition}[theorem]{Proposition}

\begin{document}

\begin{frontmatter}

\title{Robust data-driven model-reference control of linear perturbed systems via sliding mode generation} 
\thanks[footnoteinfo]{This paper was not presented at any IFAC 
meeting. Corresponding author G.~Riva, Dipartimento di Elettronica, Informazione
e Bioingegneria, Politecnico di Milano, 20133 Milan, Italy.}

\author[polimi]{Giorgio Riva}\ead{giorgio.riva@polimi.it},    
\author[polimi]{Gian Paolo Incremona}\ead{gianpaolo.incremona@polimi.it},     
\author[polimi]{Simone Formentin}\ead{simone.formentin@polimi.it}, 
\author[unipv]{Antonella Ferrara}\ead{antonella.ferrara@unipv.it}             
\address[polimi]{Dipartimento di Elettronica, Informazione e Bioingegneria, Politecnico di Milano, 20133 Milan, Italy}                                           
\address[unipv]{Dipartimento di Ingegneria Industriale e dell'Informazione, University of Pavia, 27100 Pavia, Italy}             
       

\begin{abstract}
This paper introduces a data-based integral sliding mode control scheme for robustification of model-reference controllers, accommodating generic multivariable linear systems with unknown dynamics and affected by matched disturbances. Specifically, an integral sliding mode control (ISMC) law is recast into a data-based framework relying on an integral sliding variable depending only on the reference model, without the need of modeling the plant. The main strength of the proposed approach is the enforcement of the desired reference model in closed-loop under sliding mode conditions, despite the lack of knowledge of the model dynamics and the presence of the matched disturbances. Moreover, the conditions required to guarantee an integral sliding mode generation and the closed-loop stability are formally analyzed in the paper, remarking the generality of the proposed data-driven integral sliding mode control (DD-ISMC) with respect to the related model-based counterpart. Finally, the main practices for the data-based design of the proposed control scheme are deeply discussed in the paper, and the proposed method is tested in simulation on a benchmark example, and experimentally on a real laboratory setup. Simulation and experimental evidence fully corroborates the theoretical analysis, thus motivating further research in this direction.
\end{abstract}

\end{frontmatter}

\section{Introduction}
In a multitude of today's applicative fields, from biomedicine to economy, and in many engineering domains, from transportation to industrial processes, data  are playing a role of paramount importance. In the specific case of practical control applications, whenever the dynamics of plants cannot be easily captured by mathematical models, measurement data become instrumental to design controllers capable of making the system behave as a reference model. In this context, data-driven control approaches have established themselves for their ability to address this kind of control problems, without resorting to heuristics, supported by a solid theoretical foundation \cite{formentin2014comparison}.

Within this research line, two well-established paradigms have emerged, namely \emph{indirect data-driven control} and \emph{direct data-driven control}. The former consists of two sequential design phases, that are a system identification and then a control synthesis based on the identified model. The latter instead pursues the decision of the control action in a way compatible with the recorded data. Exhaustive surveys on both the methodologies can be found e.g., in \cite{hou2013model}.

Although both approaches have been deeply investigated and elaborated in the literature, limitations become apparent under perturbed settings affecting their outcomes due to inconsistent uncertainty representations. In indirect approaches, indeed, the model identification, under perturbed conditions, can be of limited relevance; in direct methods, instead, addressing this shortcoming by learning the control policy from data can be analytically more complex and computationally burdensome \cite[Sec. 1]{dorfler2022bridging}. Bridging direct and indirect data-driven control \cite{dorfler2022bridging} addresses these problems by trading off the identification and the control decision seeking.

In particular, these findings resonate with data-driven methods based on the \emph{model-reference} paradigm, where the controller tuning problem is formulated as an identification problem based on input and output data collected from the plant and on a target closed-loop behavior, typically represented by a transfer function, dubbed \emph{reference model} \cite{campestrini2017data}. Specifically, despite direct data-driven model-reference methods are preferable in reducing the bias of identified controllers, they lack in the variance trade-off. This occurs especially when model classes are not finely chosen, which is a common situation in practical applications \cite{dorfler2022bridging}.

\subsection{From peculiar robust data-driven control to \\ robustification via sliding mode generation}

In the above mentioned \emph{model-reference} design framework, direct data-driven controller tuning methods aim at configuring a feedback controller directly from input–output data without computing an explicit plant model, in order to match a user defined closed-loop behaviour. Among the classical approaches, virtual reference feedback tuning (VRFT) uses a single batch of open-loop data to solve an offline optimization that matches a prescribed reference model with the actual closed-loop behavior, minimizing a “virtual” error between measured and ideal outputs~\cite{Campi2002,Formentin2012}. Similarly, iterative feedback tuning (IFT) refines controller parameters through repeated closed-loop experiments, estimating gradients of a performance index directly from data trajectories~\cite{Hjalmarsson1998}. Fictitious reference iterative tuning (FRIT) constructs a fictitious reference from one experiment and optimizes controller parameters offline to better match a reference model, which can improve performance in non-ideal dynamics compared to VRFT at the cost of increased computation~\cite{Soma2004}. Correlation-based tuning (CbT) enforces statistical orthogonality conditions between error and filtered signals to achieve model-reference matching in a least-squares sense, often yielding unbiased controller estimates when noise has zero mean~\cite{Karimi2004}.\\
Recent literature has continued this line of work while addressing practical performance and robustness challenges inherent in these methods. For example, auto-tuning of reference models seeks to jointly optimize the reference and controller parameters directly from data to improve achievable closed-loop performance and handle non-minimum-phase dynamics~\cite{Masti2023}, and some recent work proposes novel strategies that enforce closed-loop stability in the presence of noise and measurement limitations~\cite{breschi2021direct,Formentin2023}. Another recent activity extends robust data-driven iterative control strategies to linear systems with bounded disturbances, designing controllers that account for uncertainty directly from data by considering all systems compatible with the collected dataset, highlighting that disturbance presence fundamentally limits identification of a “true” system solely from finite data~\cite{HuLiu2024}. A related paper proposes a robust data-driven control framework tailored for probabilistic systems with aleatoric uncertainty, combining scenario optimization with direct data-driven design to achieve probabilistic stability guarantees across variations in dynamics, emphasizing that robustness to system variations must be explicitly designed rather than automatically ensured by model-free tuning~\cite{VonRohre2025}.\\
A critical issue across all these approaches is their handling of disturbances during operational control (i.e., once the controller is deployed). While many methods can mitigate the effect of noise in the training data (e.g., via instrumental variables or filtering), none of them is taylored to eliminate the effect of unknown disturbances on the closed-loop performance once the controller is fixed: the designed controller, even if tuned to a reference model on clean or carefully filtered data, will still experience residual performance degradation under unmodeled disturbances or plant variations in operation. This is because direct data-driven tuning enforces model matching under the assumption that the closed-loop behavior approximates the desired model, but external disturbances and plant variation act as unmodeled inputs that cannot be fully “learned away” from fixed amounts of finite data and thus influence the realized closed-loop output relative to the reference. This is true also for reformulation of classical methods more oriented to disturbance rejection like \cite{eckhard2018virtual}, as the focus therein is to match some specific sensitivity functions without preserving the reference-to-output matching objective of the original VRFT. Such limitations are intrinsic to model-free strategies that do not explicitly incorporate robust performance criteria or disturbance attenuation objectives in the synthesis.

As a valid alternative to tackle challenges induced by disturbances, sliding mode control (SMC) has gained widespread adoption, owing to its capability to deliver simple and computationally efficient solutions to a broad range of control problems under uncertainty conditions \cite{incremona2019smc}. In classical SMC, suitably defining the so-called sliding variable, the control law is a discontinuous function of time, aimed at steering such a variable to zero in finite time. Whenever a sliding mode is enforced after a finite transient interval, then the resulting state dynamics is proved insensitive to the so-called matched disturbances (see \cite[Ch. 1]{incremona2019smc} for a deeper discussion on classical SMC law). Moreover, SMC laws can be considered \emph{model-independent}, only relying on the knowledge of the bound on the uncertainties.

The integral sliding mode control (ISMC) proposed by \cite{utkin1996ism} can be classified instead as a \emph{model-based} approach. Its control law consists of two components, that is an ideal control law aimed at stabilizing the unperturbed closed-loop system in compliance with prescribed requirements, and a discontinuous law  taking care of the matched disturbance rejection. Specifically, assuming to know the model of the system, the ISMC law is capable of making the closed-loop system dynamics ideally coincide with that enforced via the ideal stabilizing control law  since the initial time instant, hence enhancing the robustness of the controlled system. Taking this perspective, the ISMC \cite{utkin1996ism} can be conceived as a robust model-reference controller. However, the requirement on the knowledge of the model prevents its application to model-unknown systems and its combination with data-driven controllers. 

To address this problem, recently, data-based versions of ISMC have been introduced in the literature. Among \emph{indirect} methods are ISMC
laws that rely on the use of neural-networks. In \cite{sacchi2024nnsmc} deep neural networks are exploited to approximate the model of nonlinear systems where the drift term and the matrix multiplying the input are unknown. Such an approximation is then adopted as model to design the integral sliding variable and the SMC law, as well as the ideal stabilizing law. The sequential model identification, which may not be free from an unknown approximation error, and the application of the ISMC conceived as in \cite{utkin1996ism} can however prevent from making the closed-loop system behave as the desired reference model. An overview on these neural-network based ISMC approaches can be found in \cite{ferrara2024survey}. 

Within the category of \emph{direct} approaches, the most notable is \cite{samari2025ddismc}, where a data-driven integral sliding mode control (DD-ISMC) is proposed for large-scale networks of model-unknown continuous-time nonlinear systems, exploiting two input-state trajectories to design the ideal stabilizing law and the SMC component, proving robust global asymptotic stability of the closed loop network. In this work all the state vector is assumed available, and the input-output data are retrieved in a disturbance-free setting. Moreover, such a DD-ISMC is not formulated as a model-reference control, but a stabilization control problem is addressed without requiring to make the closed loop system behave as a desired reference model. Potentially recasting this problem in a model-reference setting, under the control structure proposed in \cite{utkin1996ism}, however, the disturbance-free assumption for the collected data can prevent, in a real perturbed scenario, from making the closed loop system coincide with the reference model. 

In the same category of direct approaches, other relatively few works have proposed a comprehensive methodological framework, see e.g., \cite{liu2019ddsmc,corradini2022ddsmc,lan2025ddsmc}. However, these contributions adopt classical SMC laws rather than an ISMC method, hence being not suitable to handle model-reference control problems.
Differently from the aforementioned works which do not explicitly employ a reference model concept, \cite{riva2024ddsmc} proposes a data-driven model-reference control based on VRFT approach, combined with an ISMC for single-input single-output (SISO) continuous-time linear systems perturbed by matched disturbances.

Taking this line of research, that of exploring robust model-reference based controllers for perturbed model-unknown systems, we provide  here a comprehensive methodological framework for the design of a model-reference control  made robust through a DD-ISMC, the tuning of which is in turn grounded in data and  closed-loop reference model.

\subsection{Main contributions of the paper}
In this article, we contribute to the literature on robust data-driven model-reference control by proposing a holistic methodology which exploits a DD-ISMC to reject matched disturbances affecting the system and possible residual disturbances deriving from a suboptimal ideal action policy achieved from input-output data. 
Differently from \cite{riva2024ddsmc}, where the technical feasibility of combining ISMC and VRFT approach was investigated in a SISO setting, this paper introduces a deep  and rigorous general theory to design robust data-driven model-reference control schemes via sliding mode generation. 

We focus here on generic multivariable continuous-time linear model-unknown systems with output feedback and perturbed by matched disturbances. 
Then, motivated by the limitations of aforementioned DD-ISMC methods in the presence of disturbances, this paper proposes the design of an integral sliding variable depending exclusively on the reference model and independent of state measurements. Such a choice is proved to be equivalent to the design in \cite{utkin1996ism} when the system model is known or when no disturbances affect the data collection phase. It is shown that the resulting DD-ISMC  law enforces the reference model in closed-loop regardless of the presence of matched disturbances and of disturbances induced by the ideal stabilizing law identified from data.
Required conditions ensuring closed-loop stability in sliding mode are derived, together with a proof of the existence of an integral sliding mode and a characterization of the residual disturbance associated with the data-driven ideal model-reference control law.

Moreover, practical tuning guidelines for the proposed approach are provided by exploiting a continuous-time multivariable VRFT methodology, along with conditions for tuning the discontinuous component of the control law. In particular, by leveraging the Petersen’s lemma reformulated within the proposed framework, the gain and the matrix required to construct the discontinuous component of the proposed DD-ISMC guaranteeing the existence of a sliding mode from the initial time instant are synthesized. The theoretical results are validated through both simulations and experimental tests on a numerical benchmark and on a real-setup to further assess its feasibility in practice.

Finally, it is worth highlighting that the contribution of this work is twofold. First, it provides a rigorous and readily implementable solution to the open problem of robust data-driven model-reference control. Second, it extends the classical ISMC paradigm in its model-reference formulation to the case of model-unknown systems and in the output-feedback fashion.

\subsection{Outline of the paper}
The remainder of the paper is structured as follows. Section \ref{sec:pf} introduces the main notation, defines the considered class of systems and its assumptions, recalls some preliminaries on ISMC, and states the control problem. Section \ref{sec:ddismc} presents the proposed DD-ISMC scheme, and provides the main theoretical results. Section \ref{sec:practices} discusses how to tune the DD-ISMC law, and Section \ref{sec:num_exp_results} is devoted to the numerical and experimental assessment. The paper ends with some conclusions in Section \ref{sec:conclusions}, and an in-depth treatment of the continuous-time multivariable VRFT algorithm in Appendix \ref{app:vrft}.

\section{Problem Formulation}\label{sec:pf}
To move towards the proposed theory of robust model-reference control via sliding mode generation, we need first to introduce the
formal setup and the control problem in this section, after defining the notation adopted throughout the paper.

\subsection{Notation}

\emph{Matrices.}
The transpose of a matrix $A$ is denoted by $A'$, its trace as $\Tr(A)$, whereas $\bar{\lambda}(A)$ and $\underline{\lambda}(A)$ indicate its maximum and minimum eigenvalue, respectively. The norm of a matrix $A$ is $\norm{A}=\sqrt{\overline{\lambda}\left(A'A\right)}$. Given two matrices $A\in\R^{m\times n}$ and $B\in\R^{p\times q}$, their Kronecker product is denoted as $A\otimes B\in\R^{pm\times qn}$, whereas the vectorization operator of $A\in\R^{m\times n}$ is defined as $\vecto(A)\in\R^{mn}$. The matrix $I_k$ denotes the $k\times k$ identity matrix. Given a real matrix $A\in \R^{n\times n}$, it is called positive (negative) definite if, for all non-zero vectors $x\in \R^n$, $x'A x > 0$ ($x'A x < 0$), and it is denoted as $A>0$ ($A<0$). Hence, we do not restrict the definition of positive (negative) definite matrices to symmetric ones. Nevertheless, it holds that a real matrix $A\in \R^{n\times n}$ is positive (negative) definite if and only if its symmetric part $0.5(A+A') >0$ ($0.5(A+A')<0$). 

\emph{Signals, norms and LTI systems.} We adopt lowercase letters to denote signals, and uppercase letters to denote both matrices and LTI systems, i.e. $\mG:u\mapsto y$ is the linear mappings between input and output signals $y=\mG u$. Signals $x(t)$ in Laplace domains and the transfer function of an LTI system $\mG$ are normally indicated as $X(s)$ and $G(s)$, respectively. By an abuse of notation, which is however common in the literature, the input-output map is also indicated as $y(t)=G(s)u(t)$. The frequency response associated with $G(s)$ is denoted by $G(\jO)$.  
A signal $x$ belongs to $\mL_2$ iff its squared 2-norm, defined as $\lVert x\rVert_2=\int^{+\infty}_{-\infty} x'(t)x(t) \text{d} t$, is finite, whereas the 2-norm a vector $w$ is denoted as $\lVert w\rVert$. 
Given an LTI system with transfer function $G(s)$, its squared 2-norm is computed as $\lVert G(s)\rVert_2^2  = \int^{\infty}_{-\infty} \Tr\left( G(-\jO)' G(\jO)\right) \text{d}\omega$.
 Given a signal $x(t)\in\mathcal{L}_2$, if the derivative (in the weak sense of the theory of distributions, see e.g., \cite[Chapter 1]{adams2003sobolev}) of order
 $i$ of $x(t)$, written $x^{(i)}(t)$, is in $\mathcal{L}_2$ for all $i\leq p$, then $x(t)$ is said to
 belong to $\mathcal{H}^p$, the Sobolev space of order $p$ \cite{adams2003sobolev}. 
Given a signal $x(t):\R\mapsto\R$, the function $\mathrm{step}(x)$ is defined such that $\mathrm{step}(x)=1$ if $x>0$, and $\mathrm{step}(x)=0$ if $x\leq0$.
Given a signal $x(t):\R\mapsto\R^m$, we denote by $x_{[0,(T-1)\tau]}$, where
$\tau \in \R_{>0}$, $T\in\mathbb{N}_{\geq 1}$, the restriction in vectorized form of $x$ to the
sampled interval $[0,(T-1)\tau] \cap \mathbb{Z}$, namely $x_{[0,(T-1)\tau]} \coloneq \smat{
x(0)' & x(\tau)' & \dots & x((T-1)\tau)'}'$. Letting $i\in\mathbb{N}$ with $i\in[1,T]$, we denote also the Hankel matrix associated to $x_{[0,(T-1)\tau]}$ as
$$
X_{\{0,i,(T-i+1)\tau\}} \coloneqq 
\begin{bmatrix}
    x(0)          & x(\tau)  & \dots   & x((T-i)\tau) \\
    x(\tau)       & x(2\tau) & \dots   & x((T-i+1)\tau) \\
    \vdots        & \vdots   &  \ddots & \vdots \\
    x((i-1)\tau)  & x(i\tau) & \dots   & x((T-1)\tau)
\end{bmatrix}.
$$

\subsection{Nominal dynamics}\label{sec:nom_dyn}
Consider a strictly proper continuous-time linear time-invariant  (LTI) system $\mP:u\mapsto y$  with control input $u\in\R^m$ and a measured output $y \in\R^m$. The state realization of the plant dynamics is described by 
\begin{equation}\label{eq:ssNom}
\mP:
\begin{cases}
\dot{x}(t) = Ax(t) + Bu(t)\\
y(t) = Cx(t),
\end{cases}
\end{equation}
with state $x\in\R^n$ and transfer function given by 
\[
P(s) \coloneq C(sI-A)^{-1}B\,,
\]
where the tuple $(A,B,C)$ is unknown.
Let the measured output be $y\in\R^{m}$, 
meaning that the regulated channel $u\mapsto y$ is neither under-actuated, which is necessary for any regulator problem, nor has redundancies.
Then, let the error signal 
\begin{equation}\label{eq:e}
e(t) \coloneq r(t)-y(t)\in\R^m\,,
\end{equation}
where $r(t)\in\R^m$ is the reference signal. 

Now, consider a reference model $\M: r\mapsto \yo$ corresponding to the \emph{ideal} evolution of the closed-loop system. Moreover, we restrict the family of reference models to those described by transfer function $M(s)$.
Assume that there exists an \emph{ideal} control action $u\coloneqq \uid$ such that the output of system \eqref{eq:ssNom} under the control $\uid$ corresponds to the output of the reference model $\yo$, with $\yo(0)=C\xo(0)$, where $\xo$ as the state of the system under the control action $\uid$.

Assume also that there exists an \emph{ideal controller} $\mRid:e\mapsto \urm$ with transfer function $\Rid(s)$ and state-space realization given by the tuple $\left(\Arid,\Brid,\Crid,\Drid\right)$, capable of making the closed-loop transfer function behave as $M(s)$  of the reference model $\M$. Note that $\urm$ does not necessarily coincide with $\uid$, because of the non-uniqueness of the state realization of the reference model $\M$, as will be discussed in Section \ref{sec:mr_control} in the case of unknown model of the system. Specifically, $\urm$ coincides with $\uid$ in case of zero initial conditions.

\subsection{Perturbed dynamics and problem statement}
The real systems like \eqref{eq:ssNom} are, however, typically affected by unavoidable modeling uncertainties and external disturbances, and this consideration can be formulated by introducing the perturbed system $\mP_d:u\mapsto y$ as
\begin{equation}\label{eq:ssPer}
\mP_d:
\begin{cases}
\dot{x}(t) = Ax(t) + Bu(t)+Bd(t)\\
y(t) = Cx(t),\\
\end{cases}
\end{equation}
where $d\in\R^m$ is the whole perturbation of \emph{matched} type. Therefore,  we address the following problem.

\begin{resp}
\begin{problem}\label{pb:main}
Design a control law $u$ such that the solution of \emph{perturbed} and \emph{unknown} system \eqref{eq:ssPer} fulfills $y(t)\equiv \yo(t)$ and $x(t)\equiv \xo(t)$ for any $t\geq 0$, starting from the initial conditions $x(0)= \xo(0)$. 
\end{problem}
\end{resp}
What we really need are the following assumptions:
\begin{assumptions}
\item\label{ass:d} $\norm{d(t)}\leq \bar{d}$,  $\bar{d}=\sup_{t\geq 0}\{\norm{d(t)}\}$ known, $d\in\mathcal{H}^p$ of suitable order $p$. 
\item\label{ass:cb} $\rank(CB)=m$.
\end{assumptions}

Assumption \Ass{\ref{ass:d}} does not imply any restriction on the class of disturbances, which are commonly bounded for physical reasons, and their bounds can be characterized through a careful analysis of the considered system, relying on, e.g., data analysis or physical insights.
Assumption \Ass{\ref{ass:cb}} says that, although $C$ and $B$ are unknown, there exists $(CB)^{-1}$, which is a sufficient condition for the existence of $\uid$. Moreover, such a condition  simplifies the design of a sensible reference model, as stated in the following remark.
\\

\begin{remark}[Reference model design and $R_0(s)$ existence]
    Note that \Ass{\ref{ass:cb}} implies that the system has a strict relative degree equal to one (see \cite{mueller2009normal} for further details), which in turn guarantees that for any strictly-proper reference model $M(s)$  (that is with a vector relative degree $\left(\nu_1,\dots,\nu_m\right)$ such that $\nu _i\geq 1$, $\forall i = 1,\dots,m$) a causal controller $R_0(s)$ exists, and it can be computed as $R_0(s) = (P^{-1}(s)M(s))(I-M(s))^{-1}$. For example, considering a diagonal reference model $M(s)$, this mild requirement implies that each SISO transfer function in the main diagonal should have a relative degree $\nu_i \geq 1$. 
    \Endofremark
\end{remark}

\subsection{Integral sliding mode control}\label{sec:ismcc}
Under assumptions \Ass{\ref{ass:d},\ref{ass:cb}}, a possible strategy to address the previous problem is given by the so-called ISMC \cite{utkin1996ism}. The main idea of an ISMC is to add  to the ideal control action a discontinuous control component depending on a \emph{switching function} obtained by summing a so-called sliding variable (typically equal to a linear combination of the system states) and a so-called \emph{transient function} given by an integral term which depends on the nominal dynamics of the system and on the initial conditions. 

Specifically, given the ideal control action $\uid$, the whole control input becomes 
\begin{equation}\label{eq:u}
u(t) \coloneqq \uid(t)+u_1(t)\,,
\end{equation}
where $u_1$ is designed to be discontinuous aimed at rejecting the perturbation $d$. Assuming that the system error $e$ in \eqref{eq:e} can be regarded as the so-called sliding variable $\sigma_0$, i.e., $\sigma_0=e$, the switching function is designed as
\begin{equation}
\sigma \coloneqq \sigma_0 + \zeta\in\R^m\,,
\end{equation}
where $\zeta$ induces the integral term and represents the transient function. 

To design an integral feedback such that the so-called equivalent control \cite[Sec. 2.2]{utkin1992smc}, namely $\ueq$ (achieved by solving $\dot{\sigma}=0$ with respect to the control input) compensates the disturbance, i.e., $\ueq(t)=-d(t),\,\forall\,t\geq 0$, the transient function is determined as
\begin{equation}\label{eq:z}
\dot{\zeta} (t)= CAx(t)+CB\uid(t)-\dot{r}(t),
\end{equation}
with initial condition $\zeta(0)=-\sigma_0(0)$, so that $\sigma(0)=0$, that is a sliding mode occurs starting from the initial time instant. By virtue of assumption \Ass{\ref{ass:cb}}, such a sliding mode can be enforced by applying the control law
\begin{equation}
\label{eq:u1}
u_1(t)\coloneqq \rho(CB)^{-1}\frac{\sigma(t)}{\norm{\sigma(t)}},
\end{equation}
that is the so-called \emph{unit vector} approach, which is a convenient control structure for multivariable systems \cite[Sec. 3.6]{edwards1998smc}. Making reference to \cite[Sec. 2]{utkin1996ism}, under assumptions \Ass{\ref{ass:d},\ref{ass:cb}}, by using Lyapunov arguments, the sliding mode is enforced if the gain $\rho>\sqrt{\overline{\lambda}\left((CB)'CB\right)}\bar{d}$, that is the gain should be large enough to dominate the worst realization of the disturbance.

However, the classical ISMC is based on the perfect knowledge of the nominal model $(A,B,C)$ in \eqref{eq:ssNom} of the system, and on the measurement of the state vector. Specifically, this information is required to design the ideal control law $\uid$, which determines the transient function dynamics in \eqref{eq:z}. The approach is therefore not readily extendable to include systems with unknown models, like those used in this work. 
In the following, we devise a novel procedure to design an ISMC capable to solve Problem \ref{pb:main}, relying on a data-driven model-reference ideal control, while maintaining the same robustness property of the original ISMC in \cite{utkin1996ism}.

\section{The Proposed Model-Reference-Based ISMC}\label{sec:ddismc}
In this paper, to solve Problem \ref{pb:main}, we propose a data-driven ISMC (referred as DD-ISMC) relying on a
model-reference control, based on the measurement retrieved from the plant, as indicated in the scheme in Fig. \ref{fig:ddISMC}.
\begin{figure}[!ht]
        \centering
        \includegraphics[width=0.95\columnwidth]{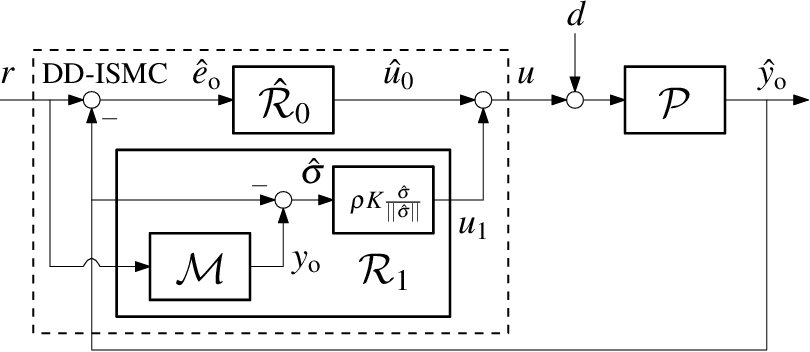}
        \caption{DD-ISMC scheme. The control signal $u(t)$ is given by the sum of the output of the ideal (data-driven) controller $\hat{\mathcal{R}}_0$ and the output of the sliding mode controller $\mathcal{R}_1$ (receiving as inputs the reference signal $r(t)$ and the controlled output $\ymdd(t)$), aimed at compensating the mismatches with respect to the reference model $\M$, including also the matched disturbance $d$ affecting the plant $\mP$.}
        \label{fig:ddISMC}
\end{figure}

\subsection{Model-reference control}\label{sec:mr_control}
Consider the following control law
\begin{equation}\label{eq:udd}
u(t)\coloneqq \udd(t)+u_1(t),
\end{equation}
where $u_1$ is a discontinuous component in the form of \eqref{eq:u1}, whereas $\udd$ plays the role of the ideal controller. Note that, such a controller could generally deviate from the ideal control action $\uid$, due to the lack of knowledge of the model. 
Specifically, $\udd$ is the output of a controller $\mRdd:e\mapsto \udd$ with transfer function $\Rdd(s)$ and state-space realization given by the tuple $\left(\Ardd,\Brdd,\Crdd,\Drdd\right)$, the state of which belongs $\R^{\ndd}$.

Therefore, the closed-loop system under the control input in \eqref{eq:udd} can be written as
\begin{equation}\label{eq:ssMdd}
\mMdd:
\begin{cases}
\xmddd(t) = \Amdd\xmdd(t)+\Bmdd r(t)+ \Budd \left( u_1(t)+d(t)\right)\\
\ymdd(t) =\Cmdd\xmdd(t),
\end{cases}
\end{equation}
where $\ymdd$ denotes the output trajectories of the system under the data-driven control $\udd$, $\Budd\coloneqq\smat{B& 0}'$, and $\Mdd\coloneq \Cmdd(sI-\Amdd)^{-1}\Bmdd$ is the corresponding transfer function, from $r$ to $\ymdd$, in terms of its state realization
\[
\Amdd\coloneqq\mmatrix[c|c]{A-B\Drdd C & B\Crdd\\\hline -\Brdd C & \Ardd},\;\Bmdd\coloneqq\mmatrix{B\Drdd \\ \Brdd},\; \Cmdd\coloneqq \mmatrix[cc]{C & 0},
\]
with $\xmdd\coloneqq\smat{\xdd &\xrdd}'$ being the 
corresponding state containing the plant state $\xdd$ under the evolution given by the data-driven  controller, with state referred to as $\xrdd$.

Having in mind the design of an ISMC, under the lack of knowledge of the nominal plant, a novel procedure to design the integral sliding variable is now discussed.
Therefore, it is instrumental to introduce the tuple $(\Am,\Bm,\Cm)$, that is one of the possible minimal state realizations associated with the desired closed-loop transfer function $M(s)$, i.e.,
\begin{equation}\label{eq:ssMddm}
\M:
\begin{cases}
\xmd(t) = \Am\xm(t)+\Bm r(t)\\
\yo(t) =\Cm\xm(t),
\end{cases}
\end{equation}
where $\xm \in \R^{\nm}$ is the corresponding state, and $\yo$ is the desired output evolution. Note that, starting from the transfer function $M(s)$, the tuple $(\Am,\Bm,\Cm)$ does not depend from the unknown plant tuple $(A,B,C)$, and it can be freely chosen among all the equivalent minimal state realizations of $M(s)$. 

Analogously to the classical definition, the sliding variable is now defined as
\begin{equation}\label{eq:sigmadd}
\sigmadd \coloneqq \sigmaddo + \zetadd\in\R^m,
\end{equation}
where $\sigmaddo=r-\ymdd=\edd$, whereas $\zetadd$ is the novel transient function, in analogy with \cite{utkin1996ism} determined by integrating
\begin{equation}\label{eq:zetadd}
\dot{\zetadd} (t) \coloneqq \Cm \Am \xm(t)+\Cm \Bm r(t)-\dot{r}(t).
\end{equation}
It is worth noticing that the new transient function depends on available quantities. 

Given the definition of the sliding variable in \eqref{eq:sigmadd}, the proposed data-based formulation of the discontinuous component $u_1$ in \eqref{eq:udd}, as the output of a regulator $\mathcal{R}_1: (r,\ymdd)\mapsto u_1$, is
\begin{equation}\label{eq:u1_dd}
    u_1(t)\coloneqq\rho K\frac{\sigmadd(t)}{{\norm{\sigmadd(t)}}},
\end{equation}
where $\rho >0$, and $K \in \R^{m\times m}$ is such that $(CB)K > 0$. Note that, since the model is unknown, the term $CB$ is not available. However, the design of the matrix gain $K$ will be addressed in Section \ref{sec:db_disc} via a data-based approach. 

We now show that the selection of the sliding variable in \eqref{eq:sigmadd} and \eqref{eq:zetadd} allows the achievement of a perfect reference model matching. To this aim, the following lemma and proposition are introduced.
\\
\begin{lemma}\label{le:sigma_eq}
Given the definition of the sliding variable in \eqref{eq:sigmadd}, where $\sigmaddo=r-\ymdd=\edd$ and the transient function is obtained by integrating \eqref{eq:zetadd}, it holds that 
\begin{equation}\label{eq:sigmadd_eq}
    \sigmadd(t)=\yo(t)-\ymdd(t)\,.
\end{equation}
\end{lemma}
\begin{proof}
The derivative of the transient function in \eqref{eq:zetadd}, by exploiting \eqref{eq:ssMddm}, can be expressed as 
\begin{equation}
    \label{eq:zetadd_eq}
    \dot{\zetadd} (t) = \yod(t)-\dot{r}(t).
\end{equation}
By integration one has
\begin{equation}
    \label{eq:zetadd_int}
    \zetadd(t) = \yo(t)-r(t)
\end{equation}
and since $\sigmaddo = r(t)-\ymdd(t)$, it implies \eqref{eq:sigmadd_eq}.

\end{proof}
\begin{proposition}\label{le:sm}
Given the plant \eqref{eq:ssPer}, such that \Ass{\ref{ass:d},\ref{ass:cb}} hold, controlled via 
\eqref{eq:udd} 
and sliding variable selected as in \eqref{eq:sigmadd} with transient function \eqref{eq:zetadd}, then, if a sliding mode $\sigmadd (t)=0,\,\forall\,t\geq 0$ is enforced, it holds that
\begin{equation}\label{eq:sm_y}
\yo(t)= \ymdd(t),\quad\forall\, t\geq 0,
\end{equation}
that is a perfect reference model matching is ensured.
\end{proposition}
\begin{proof}
From Lemma \ref{le:sigma_eq}, under condition $\sigmadd (t)=0,\,\forall\,t\geq 0$, \eqref{eq:sm_y} holds.
\end{proof}

Before delving into the analysis of the enforcement and stability of the sliding mode, it is essential to show the effects of applying a model-reference control strategy without perfect knowledge of the underlying plant. 
We can now introduce the following Lemma.
\\
\begin{lemma}\label{le:eq_system}
Given the plant \eqref{eq:ssPer}, such that \Ass{\ref{ass:d},\ref{ass:cb}} hold, controlled via $u=\udd$, it can be equivalently written as a function of the ideal control action $\uid$ as 
\begin{equation}\label{eq:ssMdd_eq}
\begin{cases}
\dot{x}(t) = Ax(t) + B\uid(t)+B(d(t)+d_0(t))\\
y(t) = Cx(t),\\
\end{cases}
\end{equation}
where $d_0(t) = \udd(t)-\uid(t)$.
\end{lemma}
\begin{proof}
Given \eqref{eq:ssMdd_eq} and substituting the expression of $d_0$, we obtain
 \begin{equation}\label{eq:ssMdd_eq_pr}
\begin{cases}
\dot{x}(t) = Ax(t) + B\udd(t)+Bd(t)\\
y(t) = Cx(t),\\
\end{cases}
\end{equation} 
that is exactly the plant \eqref{eq:ssPer}, controlled via $u=\udd$.
\end{proof}
In its simplicity, Lemma \ref{le:eq_system} provides a key interpretation of the model-reference control problem in the absence of perfect knowledge of the plant. In fact, the controlled system can be seen as the one controlled by the ideal control action $\uid$, which achieves the perfect matching of the reference model in the absence of disturbance $d$, plus an additional matched disturbance $d_0$, which accounts for the mismatch of the designed control action. Despite \eqref{eq:ssMdd_eq} representing only a fictitious dynamical system, we will use it in the following to show how the proposed sliding mode control scheme can achieve the robustness property in Proposition \ref{le:sm} for a model-reference problem, even without perfect knowledge of the plant.

As a first step, we derive the expression of the equivalent control  \cite[Ch. 2]{utkin1992smc}, the one necessary to keep the sliding variable \eqref{eq:sigmadd} at zero.
\\
\begin{proposition}\label{le:eq_control} 
Given the plant \eqref{eq:ssPer}, such that \Ass{\ref{ass:d},\ref{ass:cb}} hold, controlled via $u=\udd+u_1$, and the sliding variable in \eqref{eq:sigmadd} with transient function \eqref{eq:zetadd}, if a sliding mode $\sigmadd(t)=0$, $\forall t \geq 0$ is enforced,  the equivalent control to \eqref{eq:u1_dd} is
\begin{equation}\label{eq:u1_eq}
\ueq(t) = - d(t) - d_0(t).
\end{equation}
Moreover, the equivalent system results in being
 \begin{equation}
        \label{eq:equivalent_system}
        \begin{cases}
        \xddd(t) = A\xdd(t) + B\uid(t) \\
        \ymdd(t) = C\xdd(t).\\
        \end{cases}
\end{equation}
\end{proposition}
\begin{proof}
The proof follows directly from Lemma \ref{le:sigma_eq} and Proposition \ref{le:sm}. Indeed, under the condition of $\sigmadd(t)=0$, $\tg$, given the definition of the sliding variable \eqref{eq:sigmadd_eq}, imposing $\dot{\sigmadd}(t)=0$, $\tg$, is equivalent to $\yod(t)=\ymddd(t)$, $\tg$.
Note that, leveraging \eqref{eq:ssNom} and the definition of $\uid$, $\yod=CA\xo +CB\uid$. Moreover, exploiting plant \eqref{eq:ssPer} controlled via $u=\udd+\ueq$, and the result of Lemma \ref{le:eq_system}, one has that 
\begin{equation}
    \ymddd = CA\xdd + CB\uid+CB(d+d_0+\ueq).
\end{equation}
Hence, $\dot{\sigmadd}(t)=0$, $\tg$ implies that
 \begin{equation}\label{eq:u1eq_proof_1}
CA\xo +CB\uid = CA\xdd + CB\left(\uid +d + d_0 + \ueq \right).
\end{equation} 
Under \Ass{\ref{ass:cb}}, $\mathrm{rank}(CB)=m$, meaning that it is invertible, thus ensuring the existence and the uniqueness of the solution of \eqref{eq:u1eq_proof_1} with respect to $\ueq$.  Then, one obtains
\begin{equation}\label{eq:u1_eq_1}
\ueq(t) = - d(t) - d_0(t) -\left(CB\right)^{-1}CA\left(\xdd(t)-\xo(t)\right).
\end{equation}
By substituting \eqref{eq:u1_eq_1} into \eqref{eq:ssPer} we obtain the equivalent system
\begin{equation}
        \label{eq:equivalent_system2}
        \xddd(t) = A\xdd(t) + B\uid(t)  -B\left(CB\right)^{-1}CA\left(\xdd(t)-\xo(t)\right).
\end{equation}
Now, exploiting the dynamics of the nominal system \eqref{eq:ssNom} controlled via $\uid$, i.e., $\dot{x}_\mathrm{o}(t)=A\xo(t)+B\uid(t)$, subtracting it to \eqref{eq:equivalent_system2} and defining $\tilde{x}(t)=\xdd(t)-\xo(t)$, leads to
\begin{equation}
    \label{eq:xtilde_dyn}
    \dot{\tilde{x}}(t)=\left(I-B\left(CB\right)^{-1}C\right) A \tilde{x}(t)
\end{equation}
Since $\tilde{x}(0)=0$, this implies that $\tilde{x}(t)=0$, $\tg$. As a consequence, the equivalent control becomes \eqref{eq:u1_eq}, and the equivalent system results in \eqref{eq:equivalent_system}, which corresponds to the nominal plant \eqref{eq:ssNom} controlled via $\uid$.
\end{proof}
\begin{remark}[Equivalence with model-based ISMC]\label{rem:1}
    It is worth highlighting that the equivalent control in \eqref{eq:u1_eq} resembles that in \cite[Equation 9]{utkin1996ism}, and they coincide when $d_0=0$, that is when data-driven controller coincides with ideal control action $\uid$.\Endofremark
\end{remark}

Now, what is needed is to study the stability of the closed-loop system under the sliding mode conditions and to prove that the proposed data-driven controller enforces a sliding mode.

\subsection{Stability of the closed-loop system}
In this section, we devise the conditions required to guarantee the stability of the system in sliding mode, i.e., the properties of the poles of the system in sliding motion. Our main result is grounded on a fundamental result in \cite[Sec. 3.4]{edwards1998smc}, which is briefly recalled hereafter.

Given a general dynamical system of the form
\begin{equation}\label{eq:plant_spurg}
    \begin{cases}
\dot{x}(t) = \Phi x(t) + \Gamma u(t)\\
\sigmadd(t) = S x(t)\,,
\end{cases}
\end{equation}
and a discontinuous control action capable of guaranteeing a sliding motion, i.e., such that $\sigmadd(t)=0$, the following result can be stated. We remark that the triple $\left(\Phi,\Gamma,S\right)$ in the following lemma corresponds to the triple $\left(A,B,S\right)$ in the notation in \cite[Sec. 3.4]{edwards1998smc}.
\\
\begin{lemma} \label{lemma:spurg}(Proposition 3.4 in \cite{edwards1998smc})
    The poles of the sliding motion are given by the invariant zeros of the system triple $\left(\Phi,\Gamma,S\right)$.
\end{lemma}
As a result of this lemma, the stability property in sliding motion can be checked by looking at the invariant zeros of the plant, when the output $\sigmadd$ is considered. We now re-frame this condition in our scenario, leading to the following result.
\\
\begin{theorem}\label{th:stab}
Given the plant \eqref{eq:ssPer}, such that \Ass{\ref{ass:d},\ref{ass:cb}} hold, controlled via \eqref{eq:udd}, and a sliding variable selected as in \eqref{eq:sigmadd}, with transient function \eqref{eq:zetadd}, then, if a sliding mode $\sigmadd (t)=0,\,\forall\,t\geq 0$, is enforced, the poles in sliding motion are given by the ensemble of the invariant zeros of the plant \eqref{eq:ssPer}, the poles of the data-driven controller $\Rdd(s)$, and the poles of the desired reference model $M(s)$.
\end{theorem}
\begin{proof}
In order to prove the theorem, combine systems \eqref{eq:ssMdd} and \eqref{eq:ssMddm} in order to write the following augmented system:
\begin{equation}\label{eq:aug_sys}
    \begin{cases}
        \xad(t) = \Aa \xa(t) + \BaU \left(u_1(t)+d(t)\right) + \BaR r(t) \\
        \sigmadd(t) = \Ca \xa(t),
    \end{cases}
\end{equation}
where $\xa\coloneqq\smat{\xdd&\xrdd&\xm}'$ is the augmented state vector obtained coupling the plant $\mP_d$ in \eqref{eq:ssPer}, the controller $\mRdd$, and the desired reference model $\M$, respectively. Moreover, the tuple $(\Aa,\BaU,\BaR,\Ca)$ is defined as follows
\begin{equation}\label{eq:AugMatr}
    \begin{array}{ll}
         \Aa = \begin{bmatrix}A-B\Drdd C & B\Crdd & 0 \\ -\Brdd C & \Ardd & 0 \\ 0 & 0 & \Am \end{bmatrix} & \BaU = \begin{bmatrix}
            B \\
            0 \\
            0
        \end{bmatrix} \\
         \BaR = \begin{bmatrix} B\Drdd \\ \Brdd \\ \Bm \end{bmatrix} & \Ca =  \begin{bmatrix}-C & 0 & \Cm \end{bmatrix},
    \end{array}
\end{equation}

and the triple $(\Aa,\BaU,\Ca)$ corresponds to $(\Phi,\Gamma,S)$ in \eqref{eq:plant_spurg}. Under sliding mode condition, when the disturbance $d(t)$ is ideally compensated, we can apply Lemma \ref{lemma:spurg} to \eqref{eq:aug_sys}, meaning that the poles in sliding motion are the invariant zeros of $(\Aa,\BaU,\Ca)$.
By definition, see e.g., \cite[Lemma 2.9]{colaneri1997control}, the invariant zeros are given by
\begin{equation}
    \label{eq:inv_zeros}
    \left\{ z \in \mathbb{C}: Q(z) \text{ loses normal rank}\right\}, 
\end{equation}
where $Q(z)$ is the Rosenbrock's system matrix defined as 
\begin{equation}
    \label{eq:ros_matrix}
    Q(z) = 
    \begin{bmatrix}
        zI-\Aa & \BaU \\
        -\Ca   &  0
    \end{bmatrix}.
\end{equation}
Expanding \eqref{eq:ros_matrix}, one has
\begin{equation}
    \label{eq:ros_matrix_exp}
    Q(z) = 
    \mmatrix[ccc|c]{
        zI-(A-B\Drdd C) & -B\Crdd    & 0        &  B \\
        \Brdd C         & zI - \Ardd & 0        &  0 \\
        0               & 0          & zI-\Am  & 0 \\
        \hline
        C               & 0          & -\Cm    & 0
        }.
\end{equation}

In the following steps, we exploit classical linear algebra results, that is the rank of a matrix is invariant with respect to elementary transformations (see \cite[Ch. 5]{ayres1962matrices}). Hence, the rank of matrix $Q(z)$ in \eqref{eq:ros_matrix_exp} is equivalent to that of matrix 
\begin{equation}
    \label{eq:ros_perm1}
    Q_1(z) = \Upsilon_{1} Q(z) \Upsilon_1,
\end{equation}
where $\Upsilon_1$ is the permutation matrix
\begin{equation}
    \label{eq:perm_ros_1}
    \Upsilon_1 = \mmatrix[cccc]{
        I_n & 0 & 0 & 0 \\
        0 & I_{\ndd} & 0 & 0 \\
        0 & 0 & 0 & I_{\nm} \\
        0 & 0 & I_{m} & 0},
\end{equation}
leading to
\begin{equation}\label{eq:ros_Q1}
    Q_1(z) = 
    \mmatrix[ccc|c]{
        zI-(A-B\Drdd C) & -B\Crdd    & B        &  0 \\
        \Brdd C         & zI - \Ardd & 0        &  0 \\
        C               & 0          & 0        &  -\Cm \\
        \hline
        0               & 0          & 0         & zI-\Am
        }.
\end{equation}
Now, matrix \eqref{eq:ros_Q1} is an upper block triangular matrix, meaning that it loses rank when its two diagonal blocks do. Hence, given that the lower diagonal block is equal to $zI-\Am$, it means that the eigenvalues of the reference model are invariant zeros of \eqref{eq:aug_sys}, thus, exploiting Lemma \ref{lemma:spurg}, poles in sliding motion.\\
Then, we limit the analysis to the upper diagonal block 
\begin{equation}\label{eq:ros_hat}
    Q^{(1,1)}_{1}(z) = 
    \mmatrix[cc|c]{
        zI-(A-B\Drdd C) & -B\Crdd    & B \\
        \Brdd C         & zI - \Ardd & 0 \\
        \hline
        C               & 0          & 0 
        }.
\end{equation}
Now, exploiting another elementary transformation, that is subtracting the third row block pre-multiplied by $\Brdd$ to the second row block of $Q^{(1,1)}_{1}(z)$, we get the matrix 
\begin{equation}\label{eq:ros_row1}
    \hat{Q}^{(1,1)}_{1}(z) = 
    \mmatrix[cc|c]{
        zI-(A-B\Drdd C) & -B\Crdd    & B \\
        0        & zI - \Ardd & 0 \\
        \hline
        C               & 0          & 0 
        }.
\end{equation}
Then, exploiting a similar permutation procedure to   \eqref{eq:ros_perm1}, the rank of matrix $\hat{Q}^{(1,1)}_{1}(z)$ in \eqref{eq:ros_row1} is equivalent to that of matrix 
\begin{equation}
    Q_2(z) = \Upsilon_2  \hat{Q}^{(1,1)}_{1} \Upsilon_2,
\end{equation}
where $\Upsilon_2$ is the permutation matrix \begin{equation}
    \label{eq:perm_ros_2}
    \Upsilon_2 = \mmatrix[ccc]{
        I_n & 0 & 0 \\
        0 & 0 & I_{\ndd} \\
        0 & I_m & 0},
\end{equation}
leading to
\begin{equation}\label{eq:ros_Q2}
    Q_{2}(z) = 
    \mmatrix[cc|c]{
        zI-(A-B\Drdd C) & B    & -B\Crdd \\
        C        &  0 & 0 \\
        \hline
        0               & 0          & zI - \Ardd 
        }.
\end{equation}
As before, matrix \eqref{eq:ros_Q2} is an upper block triangular matrix, meaning that it loses rank when its two diagonal blocks do. Hence, given that the lower diagonal block is equal to $zI-\Ardd$, it means that the eigenvalues of the data-driven controller are also invariant zeros of \eqref{eq:aug_sys}, thus, by virtue of Lemma \ref{lemma:spurg}, poles in sliding motion.

Finally, consider the upper diagonal block 
\begin{equation}\label{eq:ros_perm2}
    Q^{(1,1)}_{2}(z) = 
    \mmatrix[c|c]{
        zI-(A-B\Drdd C) & B \\
        \hline
        C        &  0 }.
\end{equation}
Following the same substitution procedure in \eqref{eq:ros_row1}, the second row block of $Q^{(1,1)}_{2}(z)$, pre-multiplied by $B\Drdd$, is subtracted to the first row block, leading to:
\begin{equation}\label{eq:ros_row2}
    \hat{Q}^{(1,1)}_{2}(z) = 
    \mmatrix[c|c]{
        zI-A & B \\
        \hline
        C    &  0 }.
\end{equation}
As a last step, the second row of $\hat{Q}^{(1,1)}_{2}$ is substituted with its opposite, getting matrix
\begin{equation}\label{eq:ros_row2opp}
    \Breve{Q}^{(1,1)}_{2}(z) = 
    \mmatrix[c|c]{
        zI-A & B \\
        \hline
        -C    &  0 },
\end{equation}
which is the Rosenbrock's matrix of the plant \eqref{eq:ssPer}. This result implies that also the invariant zeros of the plant are invariant zeros of system \eqref{eq:aug_sys}, that is poles in sliding motion, which completes the proof.
\end{proof}
It is worth noticing that, consistently with the sliding mode control theory \cite[Ch. 3]{edwards1998smc}, Theorem \ref{th:stab} proves that the stability of the equivalent system is dictated by the invariant zeros of the plant, which are fixed by the choice of the selected input-output. 
On the other hand, Theorem \ref{th:stab} gives a guideline on how to select the reference model $\M$, and particularly the data-driven controller $\mRdd$, which cannot be unstable.
\\

\begin{remark}[Role of $d_0$ in Theorem \ref{th:stab}]
Note that, in system \eqref{eq:aug_sys}, the disturbance $d_0=\udd-\uid$ is implicitly contained in the tuple $(\Aa,\BaU,\BaR,\Ca)$, and its expression can be found by computing the equivalent control action $\ueq$, starting from \eqref{eq:aug_sys}.
Specifically, by posing  $\dot{\sigmadd}=0$, we get the expression
\begin{equation*}
   \Ca\Aa\xa + \Ca\BaU(u_1+d)+\Ca\Br r=0,
\end{equation*}
from which, the equivalent control action results 
\begin{equation}\label{eq:u1_eq_aug}
   \ueq = -d -(\Ca\BaU)^{-1}\left(\Ca\Aa\xa+\Ca\Br r\right).
\end{equation}
Then, by direct comparison of \eqref{eq:u1_eq_aug} with \eqref{eq:u1_eq}, we obtain \begin{equation}
    \label{eq:d0_aug}
    d_0 = (\Ca\BaU)^{-1}\left(\Ca\Aa\xa+\Ca\Br r\right),
\end{equation}
as a function of $(\Aa,\BaU,\BaR,\Ca)$.\Endofremark
\end{remark}

\subsection{Convergence analysis}
We assume now that the stability conditions in sliding mode proved in Theorem \ref{th:stab} are met and that the reference signal $r$ is bounded. As a consequence, since from  \eqref{eq:d0_aug} the disturbance $d_0$ is a function of the evolution of the augmented state $\xa$ in sliding motion and of the reference $r$, the following assumption can be introduced
\begin{assumptions}
\item\label{ass:d0}  $\norm{d_0(t)}\leq \bar{d}_0$,  $\bar{d}_0=\sup_{t\geq 0}\{\norm{d_0(t)}\}$.
\end{assumptions}

Relying on the existence of $\bar{d}_0$, we are now in a position to prove the enforcement of the sliding mode under the discontinuous control law in \eqref{eq:u1_dd}.
\\
\begin{theorem}
\label{th:sm_conv}
    Given the plant in \eqref{eq:ssPer} controlled via \eqref{eq:udd} with discontinuous action as in \eqref{eq:u1_dd}, and sliding variable defined as in \eqref{eq:sigmadd} with transient function as in \eqref{eq:zetadd}, and the reference model in \eqref{eq:ssMddm}, if \Ass{\ref{ass:d},\ref{ass:cb},\ref{ass:d0}} hold, and the gain $\rho$ in \eqref{eq:u1_dd} is such that
\begin{equation}\label{eq:rho_bound}
        \rho > \frac{\sqrt{\overline{\lambda}\left((CB)'CB\right)}}{\underline{\lambda}\left((CB)K\right)}\left( \bar{d}+ \bar{d}_0\right),
\end{equation}
    then a sliding mode $\sigmadd(t)=0$ is enforced $\tg$.
\end{theorem}
\begin{proof}
The proof follows Lyapunov arguments. Indeed, by selecting the candidate Lyapunov function 
    \begin{equation}\label{eq:V}
            V(\sigmadd)=\frac{1}{2}\sigmadd' \sigmadd,
    \end{equation}
    its time derivative is $\dot{V}(\sigmadd)=\sigmadd' \dot{\sigmadd}$.
    From Lemma \ref{le:sigma_eq}, the time derivative of the sliding variable can be expressed as $\dot{\sigmadd}=\yod - \ymddd$, where $\yod=CA\xo +CB\uid$, leveraging \eqref{eq:ssNom} and the definition of $\uid$. Moreover, exploiting plant \eqref{eq:ssPer} controlled via $u=\udd+u_1$, with $u_1$ as in \eqref{eq:u1_dd}, and the result of Lemma \ref{le:eq_system}, 
$\ymddd = CA\xdd + CB\uid+CB(d+d_0+u_1$). Hence, $\tg$, it holds that
    \begin{align}\label{eq:Vdot}
           \dot{V}(\sigmadd)=&
            \sigmadd'\bigg( CA\xo +CB\uid \nonumber \\
            &\left.-\left(CA\xdd + CB\left(\uid+d+d_0+\rho K\frac{\sigmadd}{{\norm{\sigmadd}}}\right)\right)\right) \nonumber\\             %
           = & \sigmadd' \left( -CA\left(\xdd-\xo \right) -CB\left(d+d_0\right)-\rho(CB)K\frac{\sigmadd}{{\norm{\sigmadd}}}\right) \nonumber \\
           = & - \eta \norm{\sigmadd},
   \end{align}
    where $   \eta = \rho\frac{\sigmadd'(CB)K\sigmadd}{\norm{\sigmadd}^2}+\frac{\sigmadd'\left(CA\left(\xdd-\xo \right) + CB\left(d+d_0\right)\right)}{{\norm{\sigmadd}}}$.  
    Now, in order to enforce the reachability condition \cite[Ch. 7]{slotine1991applied}, we need to prove that $\eta>0$, $\tg$, which means that
   \begin{align}\label{eq:eta_p1}
        \rho\frac{\sigmadd'(CB)K\sigmadd}{\norm{\sigmadd}^2} > -\frac{\sigmadd'\left(CA\left(\xdd-\xo \right) + CB\left(d+d_0\right)\right)}{{\norm{\sigmadd}}}.
    \end{align}
    Since \Ass{\ref{ass:cb}} holds, and matrix $(CB)K$ is positive definite by construction, hence \eqref{eq:eta_p1} can be written in terms of the control gain $\rho$ as  
    \begin{align}\label{eq:eta_p2}
        \rho > -\norm{\sigmadd} \frac{\sigmadd'\left(CA\left(\xdd-\xo \right) + CB\left(d+d_0\right)\right)}{\sigmadd'(CB)K\sigmadd} = \bar{\rho}(t).
    \end{align}
    Now, to obtain the bound in \eqref{eq:rho_bound}, exploiting \Ass{\ref{ass:d},\ref{ass:d0}} we consider the worst-case realization of $\bar{\rho}(t)$, that is 
        \begin{align}\label{eq:rho_p1}
        \bar{\rho}(t) \leq & \norm{\sigmadd} \frac{\left| \sigmadd'\left(CA\left(\xdd-\xo \right) + CB\left(d+d_0\right)\right)\right|}{\sigmadd'(CB)K\sigmadd}\nonumber \\
        \leq & \frac{\norm{\sigmadd}^2\norm{CA\left(\xdd-\xo \right) + CB\left(d+d_0\right)}}{\sigmadd'(CB)K\sigmadd}\nonumber \\    
        \leq & \frac{\norm{CA\left(\xdd-\xo \right) + CB\left(d+d_0\right)}}{\underline{\lambda}\left((CB)K\right)} \nonumber\\
        \leq & \frac{\norm{CA\left(\xdd-\xo \right)} + \norm{CB\left(d+d_0\right)}}{\underline{\lambda}\left((CB)K\right)} \nonumber\\
        \leq & \frac{\norm{CA\left(\xdd-\xo \right)} + \sqrt{\overline{\lambda}\left((CB)'CB\right)}\left(\norm{d}+\norm{d_0}\right)}{\underline{\lambda}\left((CB)K\right)} \nonumber\\
        \leq & \frac{\norm{CA\left(\xdd-\xo \right)} + \sqrt{\overline{\lambda}\left((CB)'CB\right)}\left(\bar{d}+\bar{d}_0\right)}{\underline{\lambda}\left((CB)K\right)}. \nonumber\\
    \end{align}
    By virtue of the reachability condition, there exists a time instant $\overline{t}$ such that $\sigmadd(t)=0$, $\forall t \geq \overline{t}$. Since $\sigmadd(0)=0$, then $\overline{t}=0$, and $\sigmadd(t)=0$ $\tg$, which implies that $\xdd(t)=\xo(t)$, $\tg$. Then, condition \eqref{eq:rho_p1}, combined with \eqref{eq:eta_p2}, reduces to \eqref{eq:rho_bound}.
    \end{proof}

From Theorem \ref{th:sm_conv}, it is evident that the control gain $\rho$ has to be greater than the sum of two terms. The former, $\bar{d}$,  corresponds to the minimum value to enforce the sliding mode, like in the case the nominal dynamics is known, as recalled in Section \ref{sec:ismcc} for the classical ISMC \cite{utkin1992smc}. On the other hand, the price to pay due to the lack of knowledge of the nominal dynamics is the need to dominate another term, $\bar{d}_0$, depending on the reference signal $r$ and the choice of state realization of the reference model $\M$ in \eqref{eq:ssMddm}, and the actual closed-loop system $\mMdd$ in \eqref{eq:ssMdd}. This term is in general unknown, but some considerations about its contribution are worth to be discussed. It is also worth highlighting that, under the condition $K=(CB)^{-1}$ and $\bar{d}_0=0$, \eqref{eq:rho_bound} re-conduces to the threshold proved in \cite[Sec. 2]{utkin1996ism}. 
\\
\begin{proposition}\label{prop:d0_dR}
Given the plant in \eqref{eq:ssPer} controlled via \eqref{eq:udd} with discontinuous action as in \eqref{eq:u1_dd}, and sliding variable defined as in \eqref{eq:sigmadd} with transient function as in \eqref{eq:zetadd}, and the reference model in \eqref{eq:ssMddm}, if a sliding mode $\sigmadd (t)=0,\,\forall\,t\geq 0$ is enforced,
then the transfer function between the reference $r$ and the disturbance $d_0$ is given by
\begin{equation}\label{eq:d0_tf}
    d_0 = \left(\Rdd(s)-\Rid(s)\right)\left(I-M(s)\right)r.
\end{equation}
\begin{proof}
Given the definition of $d_0 = \uid - \udd$, and considering zero initial conditions, it can be written as
\begin{equation}
\label{eq:d0_R_tf}
        d_0 = \Rdd(s) \edd - \Rid(s) \eo,
\end{equation}
where $\edd = r - \ymdd$, and $\eo = r-\yo$.
By virtue of Proposition \ref{le:sm}, $\ymdd = \yo$, meaning that $\edd = \eo = (I-M(s))r$. Finally, by combining it with \eqref{eq:d0_R_tf}, the result in \eqref{eq:d0_tf} holds.
\end{proof}
\end{proposition}
It is worth noticing that, although the residual disturbance $d_0$ is unknown, Proposition \ref{prop:d0_dR} highlights its dependence on the designed data-driven controller $\Rdd(s)$.

\section{Design Practices of the Proposed DD-ISMC}\label{sec:practices}
The DD-ISMC, proposed in Section \ref{sec:ddismc}, still requires a further discussion about the design of i) the data-driven controller $\Rdd(s)$, and ii) the matrix $K$ and the gain $\rho$ in the discontinuous control component \eqref{eq:u1_dd}.

\subsection{Data-driven controller design: a VRFT approach}\label{sec:R0_vrft}
Before discussing the selected model-reference approach, it is important to highlight the general validity of the method described in Section \ref{sec:ddismc}, which is in fact independent of the adopted design procedure of $\Rdd(s)$.

Problem \ref{pb:main}
can be effectively addressed using the data-driven control approach named VRFT \cite{Campi2002,campestrini2011virtual, Formentin2019}, which is a model-reference methodology, that enables the direct identification of an optimal controller from input-output data. This approach is particularly effective for parametric control of multivariable systems, as it allows for the direct computation of the controller, parameterized by a small set of variables, without requiring explicit process model identification, which typically involves large matrices. 
Furthermore, the VRFT approach allows to enhance practically useful properties associated with the residual disturbance $d_0$. In order to introduce such properties, it is instrumental to report in the following some key results of the VRFT proposed in \cite{Formentin2019}.
Further details about the VRFT design in continuous time and for MIMO systems are reported in the Appendix \ref{app:vrft}. One can refer to Appendix \ref{app:vrft} also for the notation adopted hereafter.

In the following, two different settings are considered: the case with the disturbance $d(t)=0$ (i.e., \emph{disturbance-free} setting), and the perturbed scenario (i.e., \emph{perturbed} setting). Note that the previous two cases corresponds to the \emph{noise-free} and the \emph{noisy} settings discussed in \cite{Formentin2019}. Indeed, any matched disturbance $d(t)$ satisfying \Ass{\ref{ass:d}} can be represented by an equivalent output disturbance belonging to $\mathcal{L}_2$.
Taking into account the \emph{disturbance-free setting} the following lemma holds.
\\
\begin{lemma}
\label{lemma:vrft_1}(Theorem 1 in \cite{Formentin2019})
Assume that $\Rid(s)$ belongs to the class of controllers in \eqref{eq:R0_vrft_fr}. It holds that
$
\theta_\mathrm{vr} = \theta_0
$,
with $\theta_\mathrm{vr}$ and $\theta_0$ being, respectively, the set of parameters obtained by the VRFT approach, and the one associated with ideal controller.
\end{lemma}
By virtue of the previous lemma, the following Proposition holds.
\\
\begin{proposition}\label{prop:d0_vrft}
   Given system \eqref{eq:ssMdd_eq}, if $\theta_\mathrm{vr}=\theta_0$ and assuming zero initial conditions, then
   $
   d_0(t) = 0
   $, and the equivalent control is $\ueq(t)=-d(t),\, \tg$.
\end{proposition}
\begin{proof}
Under the assumption that $\theta_\mathrm{vr}=\theta_0$, this implies that $\Rdd(s)=\Rid(s)$.
Then, by virtue of Proposition \ref{prop:d0_dR}, from \eqref{eq:d0_tf}, it results in $d_0(t) = 0$. As a consequence, from \eqref{eq:u1_eq}, $\ueq(t)=-d(t),\, \tg$, which concludes the proof.
\end{proof}

Making reference to \cite{Formentin2019}, when instead $\Rid(s)$ does not belong to the class of controllers in \eqref{eq:R0_vrft_fr}, it is important to notice that, by suitably filtering the experimental data (see \emph{mutatis mutandis} \cite{Formentin2012} for a detailed discussion on the selection of such prefilters), the optimal controller obtained by minimizing \eqref{eq_mimo_vrft_obj} approximately resembles the solution $\theta_\mathrm{mr}$ to Problem \ref{pb:model-ref}, i.e., $\theta_\mathrm{vr} \approx \theta_\mathrm{mr}$. Notice that $\theta_\mathrm{mr}$ represents the  closest parametrization to the realization of the ideal controller $\Rid (s)$. 
It is also worth noticing that, such a result does not provide any information about the value of $d_0$, differently from Proposition \ref{prop:d0_vrft}. However, under the condition that the class of controllers is not too under-parametrized, the condition $\theta_\mathrm{vr} \approx \theta_\mathrm{mr}$ ensures in turn that the controller $\Rdd(s)$ is a valid approximation of the ideal controller $\Rid(s)$, thus implying a moderate amplitude of residual disturbance $d_0$.

Taking now into account the \textit{perturbed setting}, that is $d(t)\neq 0$, the \textit{disturbance-free} results previously discussed can be recovered, as detailed in \cite{Formentin2019}, by building on classical \emph{instrumental variable} methods \cite{soderstrom2002instrumental}.
Hence, the same considerations about the residual disturbance $d_0(t)$ are recovered as well.

The previous considerations do not provide a guess on the bound of the residual disturbance $d_0$. Such a bound is in fact crucial to tune the gain $\rho$ based on \eqref{eq:rho_bound}. Therefore, a possible way to retrieve an estimate of the bound relies on the \emph{data-based open-loop representation} discussed in \cite{de2019formulas}. Specifically, by construction $d_0=\udd-\uid$ and we define $d_\mathrm{0,v}(t) \coloneq \hat{u}_\mathrm{0,v}(t,\theta_\mathrm{vr}) - u(t)$ as the virtual disturbance obtained as the result of the optimization problem \eqref{eq_mimo_vrft_obj}. Moreover, given \eqref{eq:d0_tf}, such a virtual disturbance can be seen at the output of the dynamical system $\Rdd(s)-\Rid(s)$ fed by the virtual error $e_\mathrm{v}(t)$. Let now introduce the vectors of data collected over the time interval $[0,(\bar{T}-1)\tau]$, where $\bar{T}\in \mathbb{N}_{\geq 1}$ represents the total number of samples, and $\tau \in \R_{>0}$ denotes the sampling time, i.e., $
     d_{\mathrm{0,v},[0,(\bar{T}-1)\tau]} = 
    \smat{
        d_\mathrm{0,v}(0)' & d_\mathrm{0,v}(\tau)' & \dots & d_\mathrm{0,v}((\bar{T}-1)\tau)'}'$ and $
     e_{\mathrm{v},[0,(\bar{T}-1)\tau]} = \smat{       e_\mathrm{v}(0)' & e_\mathrm{v}(\tau)' & \dots & e_\mathrm{v}((\bar{T}-1)\tau)'}'$, 
with associated Hankel matrices $D_{\mathrm{0,v},\{0,i,(\bar{T}-i+1)\tau\}}$ and $E_{\mathrm{v},\{0,i,(\bar{T}-i+1)\tau\}}$, respectively.

Under sliding mode conditions, based on \eqref{eq:d0_tf}, the evolution of the disturbance $d_0(t)$ can be seen as the output of the system $\Rdd(s)-\Rid(s)$ fed by $\eo(t)=r(t)-\yo(t)$. Note that, since $\eo(t)$ is the ideal evolution of the error based on the reference model, it is an available signal given the reference profile $r(t)$.
Hence, making reference to \cite[Lemma 2]{de2019formulas}, under the assumption that $e_{\mathrm{v},[0,(\bar{T}-1)\tau]}$ is persistently exciting of a sufficient order, then any input-output trajectory of length $T<\bar{T}$, i.e., $e_{\mathrm{o},[0,(T-1)\tau]}$, $d_{0,[0,(T-1)\tau]}$, can be expressed as
\begin{equation}
\begin{bmatrix}
e_{\mathrm{o},[0,(T-1)\tau]} \\
d_{0,[0,(T-1)\tau]}
\end{bmatrix} = 
\mmatrix[c]{
        E_{\mathrm{v},\{0,T,(\bar{T}-T+1)\tau\}} \\
        \hline
         D_{\mathrm{0,v},\{0,T,(\bar{T}-T+1)\tau\}}        
        }g,
\end{equation}
for some $g \in \R^{\bar{T}-T+1}$. Practically, the value of $g = \bar{g}$ associated with a specific trajectory can be computed as
\begin{equation}
\bar{g} = \left(E_{\mathrm{v},\bar{T},T}'E_{\mathrm{v},\bar{T},T}+\gamma I\right)^{-1}E_{\mathrm{v},\bar{T},T}' e_\mathrm{o}(t),
\end{equation}
where 
$
E_{\mathrm{v},\bar{T},T} 
$
is a shorthand notation for $E_{\mathrm{v},\{0,T,(\bar{T}-T+1)\tau\}}$, and $\gamma \in \R_{\geq 0} $ is a regularization coefficient to deal with the presence of disturbances.
Hence, it is possible to compute a surrogate of the residual disturbance $d_0(t)$ as
\begin{equation}
\label{eq:d0_est}
\hat{d}_0(t) =  D_{\mathrm{0,v},\{0,i,(\bar{T}-i+1)\tau\}} \bar{g},
\end{equation}
thus allowing the designer to retrieve an estimate of the bound needed in \eqref{eq:rho_bound} as 
\begin{equation}
\label{eq:d0_approx}
\bar{d}_0 \approx \sup_{t\geq 0}\{\norm{\hat{d}_0(t)}\}.
\end{equation}

\begin{remark}[ISMC necessity and gain adaptation]
\label{rem:d0}
The technique borrowed by \cite{de2019formulas} does not provide an exact reconstruction of the residual disturbance $d_0(t)$, and it can be useful just for the bound estimation, rather than for a disturbance compensation. Such a compensation is indeed not at all possible, thus making the application of a sliding mode control crucial in any case. Alternatively, a solution that does not require the knowledge of $\bar{d}_0$ can be the design of adaptation mechanisms of the gain $\rho$, as e.g., those in \cite{plestan2010new} . \Endofremark
\end{remark}

\subsection{Data-based tuning of the discontinuous component}\label{sec:db_disc}
In this section, the design of the matrix $K$ adopted in the proposed data-driven control law \eqref{eq:u1_dd}, such that the  condition $(CB)K>0$ holds, is now discussed in detail. Moreover, the design of such a matrix is instrumental to provide an estimation of the bound of the control gain $\rho$ in \eqref{eq:rho_bound}.

Considering system in \eqref{eq:ssPer}, and the first-time derivative of the output $\dot{y}=CAx+CB(u+d)$, under null initial conditions, one has that in $t=0$
\begin{equation}
    \label{eq:ydot_CB_u}
    \dot{y}(0)=CB(u(0)+d(0)).
\end{equation}

Since no disturbances on the output measurement are present, and we can compute the first-time derivative of the output offline, we assume to retrieve $\dot{y}(0)$ with negligible approximation error. From \Ass{\ref{ass:cb}}, multiplying the left and right sides of \eqref{eq:ydot_CB_u} by $(CB)^{-1}$, one has
\begin{equation}
    \label{eq:ydot_CB_inv}
    (CB)^{-1}\dot{y}(0)=u(0)+d(0),
\end{equation}
where $\norm{d(0)}<\bar{d}$, by \Ass{\ref{ass:d}}. Following \cite[Sec. 3]{bisoffi2021petersen}, this assumption implies a disturbance model such that
\begin{equation}
    \label{eq:dist_model_PL}
    d(0)d(0)' \leq \bar{d}^2 I.
\end{equation}
Now, taking $N$ different realizations of the relation in \eqref{eq:ydot_CB_inv}, we define the $i$-th tuple of data $(u^{[i]}(0),\dot{y}^{[i]}(0))$, with $i=1,\dots,N$. Then, following the development in \cite[Sec. 3.1]{bisoffi2021petersen}, we can characterize, for each $i$-th tuple, the following set of matrices $Z' = (CB)^{-1}$ consistent with data, i.e.,
\begin{multline}
        \label{eq:pet_Ci_set}
    \mathcal{C}^{[i]} \coloneqq \left\{ Z': d^{[i]}(0) \in \R^m, \right. \\
    \left. Z'\dot{y}^{[i]}(0)-u^{[i]}(0) = d^{[i]}(0), \;d^{[i]}(0)(d^{[i]}(0))'\leq \bar{d}^2I \right\}.
\end{multline}
Then, set $\mathcal{C}^{[i]}$ in \eqref{eq:pet_Ci_set} can be equivalently written in term of matrix ellipsoids as
\begin{equation}
    \label{eq:pet_Ci_ellips}
    Z'A^{[i]} Z+Z'B^{[i]}+B^{[i]'}Z+C^{[i]} \leq 0,
\end{equation}
where $A^{[i]}=\dot{y}_i(0)(\dot{y}^{[i]}(0))'$, $B^{[i]} = -\dot{y}^{[i]}(0)(u_i(0))'$, and $C^{[i]} = u^{[i]}(0)(u^{[i]}(0))'-\bar{d}^2I$. The set of all matrices $(CB)^{-1}$ consistent with all $N$ data points is then
\begin{equation}
    \label{eq:ellip_inter}
    \mathcal{I}=\bigcap^N_{i=1} \mathcal{C}^{[i]}.
\end{equation}
As discussed in \cite[Sec. 3.1]{bisoffi2021petersen}, the condition for a non-degenerate ellipsoid, i.e., $A^{[i]}>0$, is never satisfied.
However, assuming that \eqref{eq:ellip_inter} is bounded, a computable over-approximation $\bar{\mathcal{I}}$ of the set $\mathcal{I}$ can be obtained following the convex optimization procedure discussed in \cite[Sec. 5.1]{bisoffi2021petersen}, leading to
\begin{equation} 
    \label{eq:I_overapp}
    \bar{\mathcal{I}} = \left\{ Z':\right.
    \left. Z'\bar{A} Z+Z'\bar{B}+\bar{B}Z+\bar{C} \leq 0 \right\}.
\end{equation}
Matrices $\bar{A}>0$, and $\bar{B}$ are obtained by firstly applying the lossy matrix S-procedure 
\begin{equation}
    \label{eq:lossy}
    \begin{bmatrix}
        -I - \sum^N_{i=1}\tau_{i}C_i & \bar{B}'-\sum^N_{i=1}\tau_{i}B_i' &  \bar{B}' \\
        \bar{B}-\sum^N_{i=1}\tau_{i}B_i & \bar{A}-\sum^N_{i=1}\tau_{i}A_i & 0 \\
        \bar{B} & 0 & -\bar{A}
    \end{bmatrix} \leq 0,\; \text{for}\;i =1,\dots,N
\end{equation}
to impose $\mathcal{I}\subseteq\bar{\mathcal{I}}$,  and then by solving the optimization problem 
\begin{equation}
    \label{eq:overapprox_problem}
    \begin{aligned}        
    &\text{minimize}\; -\text{log}\, \left(\text{det} \, \bar{A}\right) \\
    &\text{s.t.}\; \eqref{eq:lossy},\;\bar{A}>0, \; \tau_i\geq 0, \; \text{for}\;i =1,\dots,N, 
    \end{aligned}
\end{equation}
whereas $\bar{C}=\bar{B}'\bar{A}^{-1}\bar{B}-I$.
Then, expression \eqref{eq:I_overapp} can be written in a more suitable formulation, that is,
\begin{equation}
    \label{eq:I_overlapp_Iups}
    \bar{\mathcal{I}}=\left\{\left(\bar{\zeta} + \bar{A}^{-1/2}\Upsilon Q^{1/2}\right)', \; \norm{\Upsilon} \leq 1 \right\},
\end{equation}
where $\bar{\zeta} = -\bar{A}^{-1}\bar{B}$ and $Q=\bar{B}'\bar{A}^{-1}\bar{B}-\bar{C}=I$.

Now, we exploit a classical linear algebra result, that is, given a positive definite matrix, then its inverse is also positive definite {\cite{johnson1970positive}. Hence, the tuning requirement for the matrix $K$ defined in Section \ref{sec:mr_control}, i.e., $(CB)K>0$, can be guaranteed by imposing $((CB)K)^{-1}>0$, which means $K^{-1}Z'>0$. Moreover, such condition can be equivalently achieved by exploiting the symmetric part of matrix $K^{-1}Z'>0$ as 
\begin{equation}
    \label{eq:K_cond_pet}
    -(WZ' + ZW') < 0,
\end{equation}
where $W=K^{-1}$ is the matrix inverse of the gain matrix $K$. Now, given the over-approximation $\bar{\mathcal{I}}$ of the set $\mathcal{I}$ of matrices $Z'$ consistent with data, the objective becomes to guarantee condition \eqref{eq:K_cond_pet} for any matrix $Z'=(CB)^{-1}\in \bar{\mathcal{I}}$, that is
\begin{equation}
    \label{eq:K_robust_problem}
    \begin{aligned}        
    &\text{find}\; W \\
    &\text{s.t.}\; -WZ' - ZW' < 0,\; \forall Z'\in \bar{\mathcal{I}}.
    \end{aligned}
\end{equation}
By substituting expression \eqref{eq:I_overlapp_Iups} in \eqref{eq:K_robust_problem}, it can be rewritten as
\begin{equation}
    \label{eq:K_robust_problem_2}
    \begin{aligned}        
    &\text{find}\; W \\
    &\text{s.t.}\; -W(\bar{\zeta} + \bar{A}^{-1/2}\Upsilon)' - (\bar{\zeta} + \bar{A}^{-1/2}\Upsilon)W' < 0, \; \norm{\Upsilon}\leq 1.
    \end{aligned}
\end{equation}
The constraint in \eqref{eq:K_robust_problem_2} can be manipulated as 
\begin{equation}
    \label{eq:K_robust_constr}
    -W\bar{\zeta}'-\bar{\zeta}W' -W\Upsilon'\bar{A}^{-1/2}-\bar{A}^{-1/2}\Upsilon W'< 0, \; \forall\, \Upsilon: \Upsilon'\Upsilon \leq I.
\end{equation}
This formulation enables the application of the strict Petersen’s lemma in \cite[Fact 1]{bisoffi2022data}, which means that \eqref{eq:K_robust_constr} is satisfied if and only if there exits $\lambda > 0$ such that 
\begin{equation}
    \label{eq:K_robust_constr_pet_lemma}
    -(W\bar{\zeta}'+\bar{\zeta}W') + \lambda^{-1}WW'+ \lambda \bar{A}^{-1} < 0.
\end{equation}
By applying a Schur complement, problem \eqref{eq:K_robust_problem} is actually equivalent to
\begin{equation}
\label{eq:K_robust_problem_schur}
    \begin{aligned}        
    &\text{find}\; W,\;\lambda >0 \\
    &\text{s.t.}\; 
    \begin{bmatrix}
        -(W\bar{\zeta}'+\bar{\zeta}W')+\lambda \bar{A}^{-1}  & \star\\
        W  & -\lambda I  \\
    \end{bmatrix}<0,
    \end{aligned}
\end{equation}
where the gain matrix $K$ is finally obtained as $K = W^{-1}$. We conclude this design procedure with the following proposition.
\begin{proposition}
    \label{lemma:pet}
    Given the system \eqref{eq:ssPer} satisfying assumptions \Ass{\ref{ass:d},\ref{ass:cb}}, the matrix $K$ obtained solving \eqref{eq:K_robust_problem_schur} is such that for the true matrix $CB$ it holds $(CB)K > 0$. 
\end{proposition}
\begin{proof}
Given \Ass{\ref{ass:d}}, for all $N$ tuples of data, the actual disturbance realization $d^{[i]}(0)$ is precisely equivalent to the true matrix $(CB)^{-1}$, meaning that the latter is consistent with the disturbance model, thus belonging to each set $\mathcal{C}^{[i]}$. Then, since $\bar{\mathcal{I}}$ is an over-approximation of the set $\mathcal{I}$ defined in \eqref{eq:ellip_inter} as the intersection of all sets $\mathcal{C}^{[i]}$, this implies that the true matrix $(CB)^{-1} \in \bar{\mathcal{I}}$. Then, given that the solution to problem \eqref{eq:K_robust_problem_schur} guarantees that $K^{-1}Z'>0 $, $\forall\,Z:Z'\in \bar{\mathcal{I}}$, this holds also for the true matrix $Z'=(CB)^{-1}$. Hence, recalling that the inverse of a positive definite matrix is positive definite as well, it holds also that $(CB)K>0$ for the true matrix $CB$, which concludes the proof.
\end{proof}

We are now in a position to discuss how to select the gain $\rho$ in \eqref{eq:u1_dd}. Making reference to Theorem \ref{th:sm_conv}, condition \eqref{eq:rho_bound} can not be verified due to the lack of knowledge of $CB$ and $\bar{d}_0$, whereas the only known quantities are $\bar{d}$ from \Ass{\ref{ass:d}}, and $K$ as output of the optimization problem \eqref{eq:K_robust_problem_schur}. However, one can exploit the estimate provided in \eqref{eq:d0_approx} for the bound $\bar{d}_0$. As for $CB$, knowing that $\bar{\zeta}$ is the centroid of the over-approximated ellipsoid in \eqref{eq:I_overlapp_Iups} containing $((CB)^{-1})'$, corresponding to its least-squares estimate \cite[Sec. 4.2]{bisoffi2022data}, one can exploit $\overline{CB}=(\bar{\zeta}')^{-1}$ in place of $CB$. Then, the gain $\rho$ can be chosen such that $\rho>\rho_0$ with
\begin{equation}
    \label{eq:rho_est}
    \rho_0 = \frac{\sqrt{\overline{\lambda}\left(\overline{CB}'\overline{CB}\right)}}{\underline{\lambda}\left(\overline{CB} K\right)}\left( \bar{d}+ \bar{d}_0\right).
\end{equation}
It is worth highlighting that, although no knowledge of the system is available, apart from \Ass{\ref{ass:d},\ref{ass:cb}}, \eqref{eq:rho_est} provides a valid practice to satisfy the condition of Theorem \ref{th:sm_conv}.

\section{Numerical and Experimental Results}
\label{sec:num_exp_results}
In this section, the proposed DD-ISMC algorithm is assessed relying on a numerical example, and experimentally on a more realistic case study to stabilize the pitch motion of a Quanser 2-DoF helicopter\footnote{Technical details: \url{https://www.quanser.com/resource-type/technical-resources/?_products=5666}}.

\subsection{Numerical example}
\label{sec:num_example}
Consider a plant capturing the dynamics of a triple-tank system inspired from \cite{eckhard2018virtual}, and illustrated in Fig. \ref{fig:tanks}.
\begin{figure}[!ht]
        \centering
        \includegraphics[width=0.65\columnwidth]{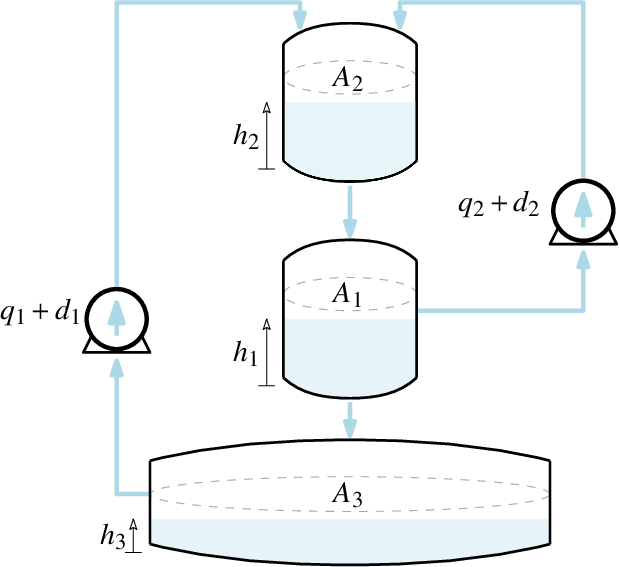}
        \caption{Schematic rendering of the triple tanks system.}
        \label{fig:tanks}
\end{figure}

Specifically, $h_1$, $h_2$ and $h_3$ are the states representing the levels of the three tanks, with corresponding cross sections $A_1,\,A_2,\,A_3$, respectively. Among the states, $h_1$ and $h_2$ correspond to the outputs, whereas the inputs are the volumetric flow rates $q_1$ and $q_2$, perturbed by the matched disturbances $d_1$ and $d_2$. 

Making reference to the plant dynamics $\mP_d$ as in \eqref{eq:ssPer}, the equations of the triple-tank system are given by
\begin{equation}\label{eq:3tanks}
\mP_d:
\begin{cases}
\dot{x}(t)\! = \!\! \begin{bmatrix}
    -\frac{k_1}{A_1} & \frac{k_2}{A_1}  & 0\\
          0          & -\frac{k_2}{A_2} & 0\\
    \frac{k_1}{A_3}  &  0               & 0
\end{bmatrix}\!x(t)\! +\! \begin{bmatrix}
    0             & -\frac{1}{A_1}\\
    \frac{1}{A_2} &  \frac{1}{A_2}\\
   -\frac{1}{A_3} & 0
\end{bmatrix}\!\!(u(t)+d(t))\\
y(t) \!=\!\! \begin{bmatrix}
    1 & 0 & 0\\
    0 & 1 & 0
\end{bmatrix}x(t),
\end{cases}
\end{equation}
where $x\coloneqq\smat{h_1 & h_2 & h_3}'$, $u\coloneqq\smat{q_1 & q_2}'$, the valve coefficients $k_1=0.008$ \si{\square\meter\per\second}, $k_2=0.0035$ \si{\square\meter\per\second}, $A_1=0.16$ \si{\square\meter}, $A_2=0.09$ \si{\square\meter} and $A_3=2.25$ \si{\square\meter}.
Note that a regulation of the output with respect to the constant levels $\bar{h}_1=0.25$ \si{\meter} and $\bar{h}_2=0.3$ \si{\meter} are considered, the corresponding  constant flow rates of which are $\bar{q}_1=0.002$ \si{\meter\cubed\per\second} and $\bar{q}_2=-0.0009$ \si{\meter\cubed\per\second}. The disturbance terms are instead given by $d_1(t)=0.0002\sin(2\pi f_1 t-\varphi_1)$ and $d_2(t)=0.0002\sin(2\pi f_2 t-\varphi_2)$, with $f_1,\,f_2$, $\varphi_1,\,\varphi_2$ being sampled for each test from an uniform distribution in the intervals $[0.005,0.05]$ \si{\hertz}, $[0,2\pi]$ \si{\radian}, respectively. The related bound is therefore $\bar{d}=0.000283$ \si{\meter\cubed\per\second}.
It is worth highlighting that the actual tuple $(A,B,C)$ is unknown, and the above parameters are indicated to emulate in simulation the real system, the model of which is assumed hidden to the control designer. 

Taking into account the design practices in Section \ref{sec:R0_vrft}, the VRFT approach is adopted to design the data-driven controller. Specifically, we performed a suitable data collection phase on the open-loop system \eqref{eq:3tanks} using the PRBS input, with amplitude \SI{0.0005}{\meter\cubed\per\second}.
The results of such a data collection phase are shown in Fig. \ref{fig:tanksPRBS}, where the employed PRBS input together with the disturbance evolution is reported on the left, while the output levels on the right.
\begin{figure*}[!t]
        \centering
        \includegraphics[width=\textwidth]{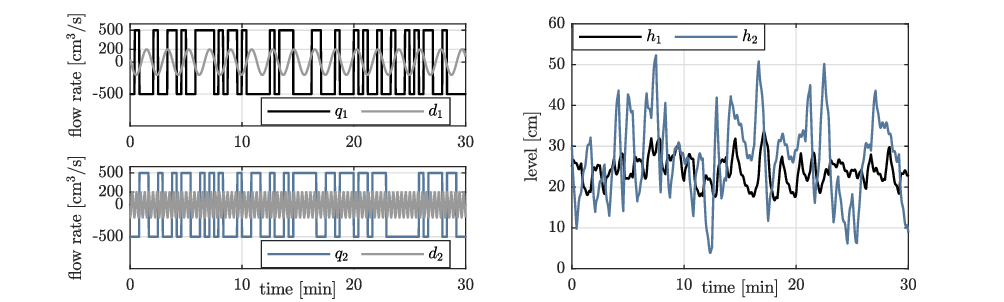}
        \caption{Open-loop input-output data collection. The PRBS input signals (plots on the left) $q_1$ (black line) and $q_2$ (blue line)  shifted by the corresponding equilibrium values $\bar{q}_1$ and $\bar{q}_2$ are fed into the plant affected by the matched disturbances $d_1$ and $d_2$ (gray lines). The corresponding output levels (plot on the right) are $h_1$ (black line) and $h_2$ (blue line).}
        \label{fig:tanksPRBS}
\end{figure*}
\begin{figure}[!t]
        \centering
        \includegraphics[width=\columnwidth]{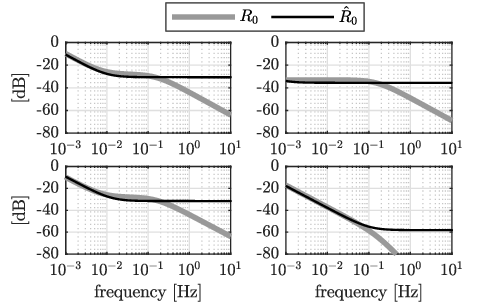}
        \caption{Bode diagrams of the magnitude associated with the frequency response of the ideal regulator $R_0(s)$ (black line) and the VRFT based controller $\hat{R}_0(s)$ (gray line) with all their components $R_{0,ij}(s)$ and $\hat{R}_{0,ij}(s)$, $i,j=1,2$, respectively.}
        \label{fig:tanksBode}
\end{figure}
\begin{figure}[!t]
        \centering
        \includegraphics[width=0.9\columnwidth]{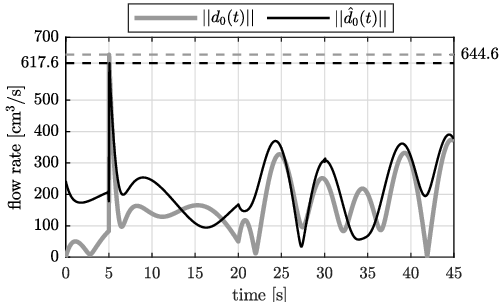}
        \caption{Estimation of the bound of the residual disturbance $d_0$. The time evolution of the norm  associated with the real residual disturbance (gray line), $\norm{d_0(t)}$, is compared with the one (black line) related to the estimation in \eqref{eq:d0_est}, $\norm{\hat{d}_0(t)}$, so that the available bound is $\bar{d}_0$ in \eqref{eq:d0_approx} (dashed black line), close to the real one (dashed gray line) which is not available in practice (see Remark \ref{rem:d0}).}
        \label{fig:tanksd0}
\end{figure}

The selected reference model for the application is given by
\begin{equation}
\label{eq:num_sim_mr}
M(s) = \begin{bmatrix}
    \frac{1}{(1+2s)^2} & 0\\
    0 & \frac{1}{(1+2s)^2}
\end{bmatrix}.
\end{equation}
As for the class of controllers, we select proportional-integral ones, which, following the expression in \eqref{eq:mimo_R_tf}, are given by
\begin{equation}
    \label{eq:sim_R0_class}
    \Rdd(s,\theta) = \begin{bmatrix}
    \frac{\theta_0 s + \theta_1}{s} & \frac{\theta_2 s + \theta_3}{s}\\
    \frac{\theta_4 s + \theta_5}{s} & \frac{\theta_6 s + \theta_7}{s}
    \end{bmatrix}=\frac{1}{s}\begin{bmatrix}
    \theta_0 s + \theta_1 & \theta_2 s + \theta_3\\
    \theta_4 s + \theta_5 & \theta_6 s + \theta_7
    \end{bmatrix},
\end{equation}
with $\theta = \smat{\theta_0 & \theta_1 & \theta_2 & \theta_3 & \theta_4 & \theta_5 & \theta_6 & \theta_7}'$ being the vector of tunable parameters. Following the VRFT procedure in Appendix \ref{app:vrft} the optimal parameters are $$\theta_\mathrm{vr}=\smat{0.029 & 0.0018 & 0.0167 & 0.00007 & -0.0265 & -0.0021 & -0.0012 & 0.0008}'.$$
Fig. \ref{fig:tanksBode} shows the Bode diagrams of the magnitude associated with the frequency response of the VRFT based controller $\hat{R}_0(s)$ (with all its components $\hat{R}_{0,ij}(s)$, $i,j=1,2$) compared with those of the components of the ideal regulator $R_0(s)$, computed from the reference model $M(s)$ and the plant transfer function $P(s)$ from \eqref{eq:3tanks} as $R_0(s) = (P^{-1}(s)M(s))(I-M(s))^{-1}$. It is apparent that, although the ideal controller $R_0(s)$ does not belong to the class of the VRFT controllers in \eqref{eq:sim_R0_class}, the VFRT procedure provides a satisfactory match between the two, particularly at low frequencies.

In order to design the SMC component, as discussed in Section \ref{sec:db_disc}, the tuning of the control gain $\rho$ and the matrix $K$ are needed. Specifically, the former depends on the bound of the residual disturbance $d_0$. Recalling that this value is associated with the closed-loop evolution of the system given the specific reference signal $r$, we selected for this analysis 
\begin{align*}
   & r_1(t)=0.25-0.015\mathrm{step}(t-5)\\
    &\quad+0.01\sin(2\pi0.1(t-20))\mathrm{step}(t-20)\text{ \si{\meter} and }\\
   & r_2(t)=0.3+0.02\sin(2\pi0.05t) \text{ \si{\meter}.}
\end{align*}
According to the procedure in Section \ref{sec:R0_vrft}, relying on \eqref{eq:d0_est}, the bound in \eqref{eq:d0_approx} is found equal to $\bar{d}_0=0.0006446$ \si{\meter\cubed\per\second} (see Fig. \ref{fig:tanksd0} where the time evolution of the norm of the real residual disturbance $\norm{d_0(t)}$ and that of the estimated one $\norm{\hat{d}_0(t)}$ are illustrated).
To find the matrix $K$, relying on the theory in Section \ref{sec:db_disc}, we performed 100 simulations with random constant inputs $q_1$ and $q_2$, and related random disturbance realizations around the constant pair values $(\bar{h}_1,\bar{h}_2)$ and $(\bar{q}_1,\bar{q}_2)$, from which the first-time derivative of the output in $t=0$ has been numerically computed. Then, from \eqref{eq:pet_Ci_ellips}--\eqref{eq:K_robust_problem_schur}, the achieved value of the matrix is
$$
K = \begin{bmatrix}
    3.9665  & 5.1174\\
    -6.4973 & 4.5234
\end{bmatrix}
$$
such that $(CB)K>0$. Moreover, as discussed in Section \ref{sec:db_disc}, the estimate of the product $CB$ is computed as
$$
\overline{CB} = 
\begin{bmatrix}
    -0.2506  & -8.8468\\
    17.7411 & 15.3833
\end{bmatrix}.
$$
Given $\bar{d}$, $\bar{d}_0$, $K$ and $\overline{CB}$, the value of the control gain as in \eqref{eq:rho_est} is $\rho=0.000477$ \si{\meter\cubed\per\second}.

\begin{figure*}[!t]
        \centering
        \includegraphics[width=0.9\textwidth]{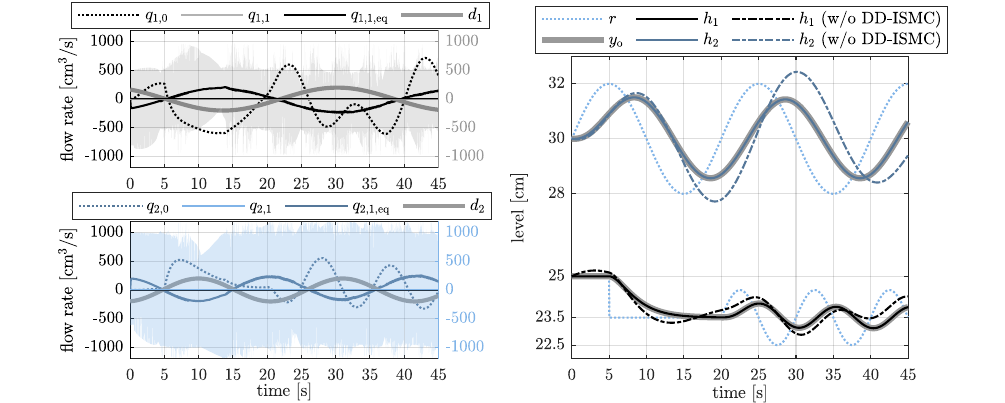}
        \caption{Closed-loop simulations when the ideal controller $R_0(s)$ belongs to the class of controllers employed in the VRFT approach (see Lemma \ref{lemma:vrft_1}). The ideal control signals $q_{1,0}$ (dotted black line) and  $q_{2,0}$ (dotted blue line) and the related discontinuous components $q_{1,1}$ (gray line) and $q_{2,1}$ (light blue line) summed up are fed into the plant. The equivalent control signals $q_{1,1,\text{eq}}$ (black line) and $q_{2,1,\text{eq}}$ (blue line) perfectly compensate the disturbances $d_1$ and $d_2$ (bold gray lines); see Proposition \ref{prop:d0_vrft} (plots on the left). Given the reference signals $r$ (dotted light blue lines), the plant outputs $h_1$ (black line) and $h_2$ (blue lines) perfectly match the outputs $\yo$ of the reference model (bold gray lines); see Proposition \ref{le:sm}. Excluding the SMC components, the outputs $h_1$ (dashed-dotted black line) and $h_2$ (dashed-dotted blue line) do not match the reference model.}
        \label{fig:tanksR0}
\end{figure*}
\begin{figure*}[!t]
        \centering
        \includegraphics[width=0.9\textwidth]{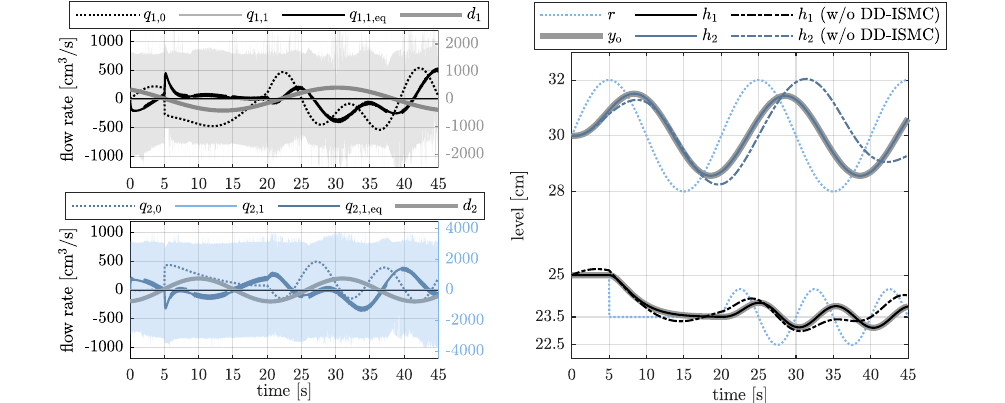}
        \caption{Closed-loop simulations when the ideal controller $R_0(s)$ does not belongs to the class of controllers employed in the VRFT approach. The ideal control signals $q_{1,0}$ (dotted black line) and  $q_{2,0}$ (dotted blue line) and the related discontinuous components $q_{1,1}$ (gray line) and $q_{2,1}$ (light blue line) summed up are fed into the plant. The equivalent control signals $q_{1,1,\text{eq}}$ (black line) and $q_{2,1,\text{eq}}$ (blue line) compensate the disturbances $d_1$ and $d_2$ (bold gray lines), together the residual disturbance $d_0$ generated by the VRFT controller $\hat{R}_0(s)$ (plots on the left). Given the reference signals $r$ (dotted light blue lines), the plant outputs $h_1$ (black line) and $h_2$ (blue lines) still perfectly match the outputs $\yo$ of the reference model (bold gray lines); see Proposition \ref{le:sm}. Excluding the SMC components, the outputs $h_1$ (dashed-dotted black line) and $h_2$ (dashed-dotted blue line) do not match the reference model.}
        \label{fig:tanksR0hat}
\end{figure*}
Fig. \ref{fig:tanksR0} and Fig. \ref{fig:tanksR0hat} illustrate the outcomes of the simulation tests. More precisely, a first test was carried out relying on the proposed DD-ISMC when the ideal controller $R_0(s)$ is used. We remark that the usage of such a controller has just a theoretical meaning rather than a practical relevance to assess the main result of the paper. One can observe that the plant outputs $h_1$ and $h_2$ match the reference model output $\yo$ even in presence of disturbances. The corresponding components  of the inputs $q_{1,1},\,q_{2,1},\,q_{1,0}$ and $q_{2,0}$  are also reported together with the equivalent control signals $q_{1,1,\text{eq}}$ and $q_{2,1,\text{eq}}$ (computed as output of a low-pass filter with suitable time constant, see \cite[Ch. 2]{incremona2019smc}). These two signals, as expected from Proposition \ref{prop:d0_vrft}, perfectly compensate only the disturbances $d_1$ and $d_2$, since the residual disturbance is $d_0=0$, thus recovering the results achievable with the model-based ISMC \cite{utkin1996ism}; see Remark \ref{rem:1}.  

The second test relies instead on the VRFT controller \eqref{eq:sim_R0_class} previously described. Also in this case the plant outputs $h_1$ and $h_2$ match the reference model output $\yo$ even in presence of disturbances. The corresponding components  of the inputs $q_{1,1},\,q_{2,1},\,q_{1,0}$ and $q_{2,0}$ are shown, but this time the equivalent control signals $q_{1,1,\text{eq}}$ and $q_{2,1,\text{eq}}$ compensate not only the disturbances $d_1$ and $d_2$ but also the residual disturbance $d_0$ generated by the mismatch between the VRFT controller and the ideal one to achieve the reference model. One can note that the effects of such a mismatch in the peak values of the equivalent control signals $q_{1,1,\text{eq}}$ and $q_{2,1,\text{eq}}$ at $t=5$ \si{\second} justified by the difference at high frequencies between the two controllers, highlighted in Fig. \ref{fig:tanksBode}. These results satisfactory confirm the developed theory.

\subsection{Experimental validation}
\label{sec:exp_val}
Now, we test the proposed control architecture on an experimental setup, that is the Quanser Aero laboratory apparatus shown in Fig.~\ref{fig:setup}.
\begin{figure}[!ht]
    \centering
    \includegraphics[width=0.9\columnwidth]{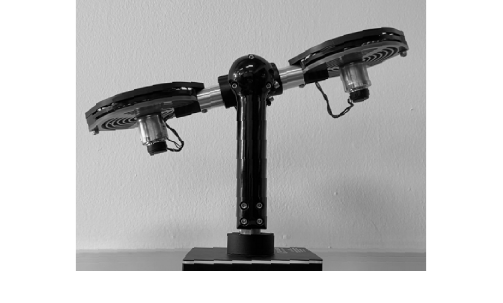}
    \caption{Quanser Aero helicopter used as the experimental setup, taking the two rotor axes parallel configuration.}
    \label{fig:setup}
\end{figure}

The Quanser Aero is a fixed-base twin-rotor helicopter emulator with two degrees of freedom (pitch and yaw) consisting of two rotors taking either a parallel axis configuration (i.e., both rotors with vertical axis at the horizontal rest position) or an orthogonal axis configuration (i.e., one rotor vertical and one horizontal at the same position).
In the considered setup, the two rotor axes are parallel (vertical at rest, as shown in Fig.~\ref{fig:setup}) and the yaw motion around the vertical axis is blocked. Specifically, we consider the special case of a SISO system, where the controlled output $y$ is the pitch rate $\dot\vartheta$, estimated based on the encoder counts, and the control input $u$ is the \emph{differential} voltage $v$ of the motor, i.e., one motor is fed with $+0.5v$ and the other with $-0.5v$, so that both voltages remain within the admissible range of $[-20\,\text{V},20\,\text{V}]$. We note that the pitch rate was selected as the controlled variable, rather than the pitch angle, because \Ass{\ref{ass:cb}} implies a relative-degree-one requirement in the SISO case. Nevertheless, the presence of non-negligible propeller dynamics increases the effective relative degree, meaning that the proposed theoretical properties hold only approximately. We also emphasize that choosing the pitch rate is not restrictive, since an outer loop regulating the actual pitch angle can be readily tuned based on the imposed reference model. The sampling time of the system is equal to $\tau=0.001$ s.

As discussed in Section \ref{sec:R0_vrft}, and analogously to Section \ref{sec:num_example}, the VRFT approach can be used to tune the controller $\Rdd(s)$ based on a desired reference model. To this end, we performed a suitable open-loop experiment on the Quanser Aero exploiting the PRBS input, with amplitude \SI{20}{\volt}, limiting the bandwidth to $10$ Hz to cope with the presence of static friction. The results of the data collection experiment are shown in Fig. \ref{fig:open_loop_exp}, where the employed PRBS input is reported on the top, while the obtained pitch rate profile on the bottom.
\begin{figure}[!ht]
    \centering
    \includegraphics[width=0.85\columnwidth]{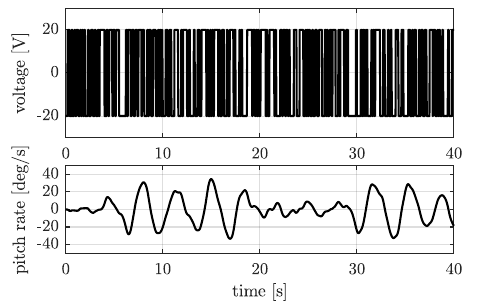}
    \caption{
    Open-loop experiments for input-output data collection. The PRBS \emph{differential} motor input voltage  is fed into the plant (plot on the top). The corresponding output is the pitch rate (plot on the bottom).}
    \label{fig:open_loop_exp}
\end{figure}

\begin{figure*}[t!]
    \centering
    \includegraphics[width=0.9\textwidth]{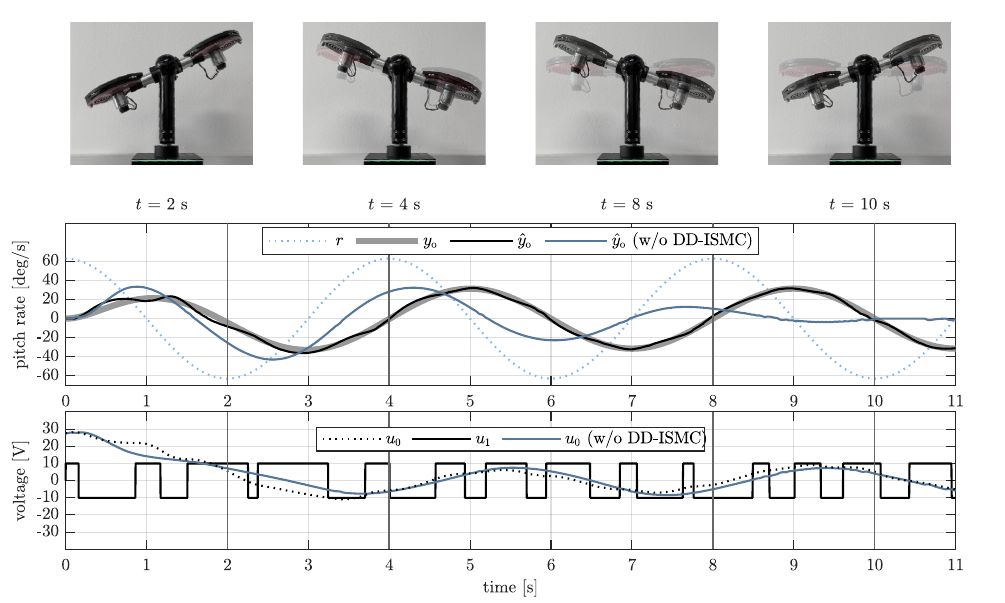}
    \caption{Closed-loop experiments when the proposed DD-ISMC scheme is applied to the Quanser Aero helicopter.  Given the reference signal $r$ (dotted light blue line), the plant output $\dot{\vartheta}$ (black line) practically tracks the output $\yo$ of the reference model (bold gray line). Excluding the SMC components, the output $\dot{\vartheta}$ (blue line) does not match the reference model. This behavior is apparent in the four snapshots at the time instants $t\in\{2,4,8,10\}$ s, where the motion generated by the DD-ISMC is overlapped to that given when the SMC component of the control law is discarded (cut-out frames). The ideal control signal $v_{0}$ (dotted black line) and the related discontinuous components $v_{1}$ (black line) summed up are fed into the plant. The corresponding ideal control signal $v_{0}$ (blue line) is also reported, when the SMC component is discarded.}
    \label{fig:exp_frames}
\end{figure*}

 We select a sensible reference model for the application, namely the second-order low-pass filter with unitary gain
\begin{equation}
    \label{eq:num_ex_mr}
    M(s) = \frac{1}{\left(1+\frac{s}{2\pi f_{\text{m}}}\right)^2}, 
\end{equation}
with $f_{\text{m}} = 0.25$ Hz. With regard to the class of controllers, inspired by prior knowledge about the system, we select
\begin{equation}
    \label{eq:exp_R0_class}
    \Rdd(s,\theta) = \theta_0 + \frac{\theta_1}{s} + \frac{\theta_2}{s^2}=\frac{\theta_0 s^2 + \theta_1 s + \theta_2}{s^2},
\end{equation}
with $\theta = \smat{\theta_0 & \theta_1 & \theta_2}'$ being the vector of tunable parameters. 
The outcome of the VRFT optimization (applying the procedure in Appendix \ref{app:vrft} recast to the SISO system framework) are the parameters $\theta_\mathrm{vr}=\smat{0.439 & 0.098 & 1.287}'$.

As for the SMC component, designed as in \eqref{eq:u1_dd}, it is worth noticing that the matrices $CB$ and $K$ are scalar in the SISO case and $K$ can be hence embedded in the control gain $\rho$. Moreover, taking into account the saturation limit of the input imposed by the adopted platform, in order to leave a certain margin of action to the ideal control component, the value of the SMC gain is selected equal to $\rho=10$ V. The reference signal is instead selected as $r(t)=20\pi\cos(0.5\pi t)$.

Fig. \ref{fig:exp_frames} shows the outcomes of the experiments. More precisely, a first experiment was carried out relying only on the ideal VRFT-based control component $v_0$, without the SMC law $v_1$. As it is apparent, the plant output $\dot\vartheta$ does not match the reference model output $\yo$ due to the effect of the effective disturbances and unavoidable modeling uncertainty (e.g., friction phenomena) affecting the laboratory setup. The second experiment was performed applying the proposed DD-ISMC scheme. As expected, the achieved results outperform the ones without SMC component. More precisely, the plant output follows in practical sense the output of the reference model, i.e., $\dot\vartheta$ evolves in a small neighborhood of the reference model signal $\yo$. This time evolution is consistent with that of the SMC signal $v_1$, which switches between $-10$ V and $+10$ V at a finite frequency like in practical SMC, due to the fact that the actual relative degree is not exactly equal to 1, for the reasons explained above. Nevertheless, this outcome must be interpreted as a promising result confirming the robustness of the proposed approach, highlighting its applicability in field implementations where the relative degree is typically not precisely known and can be potentially different from 1 because of unknown hidden dynamics.

\section{Conclusions}\label{sec:conclusions}
This paper has introduced a comprehensive methodology to design robust data-driven model-reference controllers via sliding mode generation. A DD-ISMC approach has been designed, capable of handling generic multivariable linear systems, the model of which is completely unknown and affected by matched disturbances. The main theoretical results are related to: (i) the design of an integral sliding variable depending only on the reference model to reproduce it in closed-loop (Lemma \ref{le:sigma_eq}, Proposition \ref{le:sm}); (ii) the determination of a residual disturbance, as an effect to the used data-driven model-reference controller, and of the equivalent controlled system (Lemma \ref{le:eq_system} and Proposition \ref{le:eq_control}); (iii) the determination  of the conditions required to guarantee the stability of the closed-loop system in sliding mode (Theorem \ref{th:stab}); (iv) the proof of the existence of an integral sliding mode (Theorem \ref{th:sm_conv}); and (v) the characterization of the residual disturbance (Proposition \ref{prop:d0_dR}). Furthermore, for the sake of practical implementation, a detailed discussion of the main practices has been reported to provide procedures for the design of a model-reference control relying on a VRFT approach, and for the tuning of the sliding mode control matrix and gain starting from collected data. Finally, simulation and experimental results illustrated in the paper have confirmed the theoretical analysis, showing the enforcement of an integral sliding mode, thus of the reference model in closed-loop, despite the presence of the matched disturbance and the lack of knowledge on the plant dynamics.


\appendix
{\noindent\normalfont\normalsize\bfseries\appendixname}
\section{Continuous-time MIMO VRFT}\label{app:vrft}
This section presents the main steps to design the VRFT approach in the continuous-time MIMO case. It is worth highlighting that, deriving such a formulation newly extends the one presented in \cite{Formentin2019} for the SISO case.

Considering system \eqref{eq:ssNom}, with the tuple $(A,B,C)$ fully unknown, the plant dynamics can be described by the system of ODEs
\begin{equation}
\mP_j:\sum_{k=0}^{n_{P_j}} a_{jk} y_j^{(k)}(t)-\sum_{i=1}^{m}\left(\sum_{q=0}^{m_{P_{ji}}}b_{jiq} u_i^{(q)}(t)\right)=0,
\label{eq:mimo_single_output}
\end{equation}
where $u_i(t)$, $i=1,\hdots, m$, are the input variables, whereas $u_i^{(q)}$ denotes their derivatives of order $q$, with $q=0,\hdots, m_{P_{ji}}$. Moroever, $y_j(t)$, $j=1,\hdots, m$, is the measurable output with derivatives $y_j^{(k)}$, $k=0,\hdots, n_{P_j}$, while $a_{jk},\,b_{jiq}$ are scalar plant coefficients.
Then, making reference to \cite{Formentin2019}, the following assumption holds:
\begin{assumptions}
\item $u(t)$ belongs to the same $\mathcal{H}^p$ as $d(t)$ in \Ass{\ref{ass:d}}, such that $y(t)$, and all other signals that will be derived from it, are in $\mathcal{L}_2$. \label{ass:vrft_uy}
\end{assumptions}
Taking into account the error in \eqref{eq:e}, 
under the assumption \Ass{\ref{ass:vrft_uy}}, 
the dynamics of the controller is described by the matrix of transfer functions $\Rdd(s)$ (corresponding to the ideal controller defined in Section \ref{sec:ddismc}), with the state-space representation given by the tuple $\left(\Ardd,\Brdd,\Crdd,\Drdd\right)$:
\begin{equation}
\begin{aligned}
\Rdd(s)\!=\!&
{\setlength{\arraycolsep}{1pt}\begin{bmatrix} \hat{R}_{0,11}(s)  & \cdots & \hat{R}_{0,1m}(s) \\ \vdots & \ddots & \vdots \\ \hat{R}_{0,m1}(s) & \cdots & \hat{R}_{0,mm}(s) \end{bmatrix}\!\!\! = \!\!
\frac{1}{\chi(s)}\begin{bmatrix} \hat{N}_{11}(s)  & \cdots & \hat{N}_{1m}(s) \\ \vdots & \ddots & \vdots \\ \hat{N}_{m1}(s) & \cdots & \hat{N}_{mm}(s) \end{bmatrix}}\!,
\label{eq:mimo_R_tf}
\end{aligned}
\end{equation}
where $\chi(s)=\sum^{n_{R}}_{k=0} \alpha_{k}s^k$ is the least common denominator of all $\hat{R}_{0,ij}(s)$, and $\hat{N}_{ij}(s)$ are the numerators of  $\hat{R}_{0,ij}(s)$, which can be expressed as
$
\hat{N}_{ij}(s) = \sum_{q=0}^{m_{R_{ij}}}\beta_{ijq} s^q
\label{eq:mimo_R_N}
$.
Analogously to the process model \eqref{eq:mimo_single_output}, the regulator in \eqref{eq:mimo_R_tf} can be described by the following system of ODEs
\begin{equation}
\hat{\mathcal{R}}_{0,i}(\theta):\sum_{k=0}^{n_{R}} \alpha_k \hat{u}_{0,i}^{(k)}(t)-\sum_{j=1}^{m}\left(\sum_{q=0}^{m_{R_{ij}}}\beta_{ijq} e_j^{(q)}(t)\right)=0,
\label{eq:mimo_R_ode}
\end{equation}
where we assume that the parameters $\alpha_k$ are fixed, while $\beta_{ijq}$ are tunable and contained in the vector $\theta$.
As noted in \cite{Formentin2019}, the fact that the parameters $\alpha_k$ are fixed allows us to obtain a convex formulation of the design problem. Notice also that the choice of $m_{R_{ij}}$ and
 $n_{R}$ can be either free, such that $m_{R_{ij}}\leq n_R$, or constrained by the application, and moreover they are not related to $m_{P_{ji}}$ and $n_{P_j}$ (which are unknown).
 
After a suitable manipulation of \eqref{eq:mimo_R_tf}, the transfer function matrix of the controller can be written as
\begin{equation}\label{eq:R0_vrft_fr}
\Rdd(s,\theta)= \sum^{\overline{m}_R}_{q=0} \Pi_q(\theta) \frac{s^q}{\chi(s)},
\end{equation}
where $\overline{m}_R = \max_{ij}(m_{R_{ij}})$, and $\Pi_q(\theta) \in \R^{m\times m}$, containing parameters $\beta_{ijq}$ stacked in a suitable way. 
Hence, the vector of tunable parameters $\theta$ can be  expressed as
$\theta=\begin{bmatrix} \vecto(\Pi_{0})'  \cdots \vecto(\Pi_{\overline{m}_R})' \end{bmatrix}'.$

Now, the problem of designing a controller so that the closed-loop behavior is as close as possible to the behavior of a reference model $\M$, with transfer function $M(s)$, is known as the model-reference control problem, which is here formally stated.
\begin{resp}
\begin{problem}\label{pb:model-ref}
Design a controller of the form \eqref{eq:mimo_R_ode} that minimizes the model-reference cost function
\begin{equation}
\resizebox{1.0\hsize}{!}{$
J_{\text{mr}}(\theta)= \Big\lVert W_{\text{mr}}(s)\left(M(s)-\left(I+P(s)\Rdd(s,\theta) \right)^{-1}P(s)\Rdd(s,\theta)\right)\Big\rVert_2^2 $}.
\label{eq_mimo_mr_obj}
\end{equation}
\end{problem}
\end{resp}
The cost function $J_\mathrm{mr}$ in \eqref{eq_mimo_mr_obj} penalizes the difference, weighted by $W_\mathrm{mr}(s)$, between the frequency response of the closed-loop control system, parametrized by $\theta$, and that of the reference model. We define the solution to Problem \ref{pb:model-ref} as the \emph{model-reference} controller $\hat{\mathcal{R}}_\mathrm{mr}$, based on the optimal parametrization $\theta_\mathrm{mr}=\arg \min_{\theta} J_\mathrm{mr}(\theta)$. Note that, in general, $J_\mathrm{mr}(\theta_\mathrm{mr})\neq 0$, meaning that the achieved controller $\hat{\mathcal{R}}_\mathrm{mr}\neq \mRid$, with $\mRid$ the ideal controller with transfer function $\Rid(s)$ defined in Section \ref{sec:nom_dyn}. On the other hand, when $\Rid(s)$ belongs to the class of chosen controllers in \eqref{eq:R0_vrft_fr}, we define $\theta_0$ as the corresponding set of parameters, and perfect model-reference matching is achieved, that is $\theta_\mathrm{mr}=\theta_0$.

In the case of unknown plant $P(s)$, Problem \ref{pb:model-ref} cannot be directly solved, and the VRFT approach leverages data to build an approximate solution, without estimating a model of the plant. In the following, we refer to $u(t)$ and $y(t)$ as signals recorded and stored from plant experiments. 
\\

\begin{remark}[VRFT applicability]
It is worth noticing that the VRFT approach can be applied both in the case of stable and unstable plants. As for the case of unstable plant, the identification experiments must be performed in closed-loop with a preliminary stabilizing controller \cite{Campi2002}.\Endofremark
\end{remark}

\begin{figure}[!ht]
        \centering
        \includegraphics[width=0.8\columnwidth]{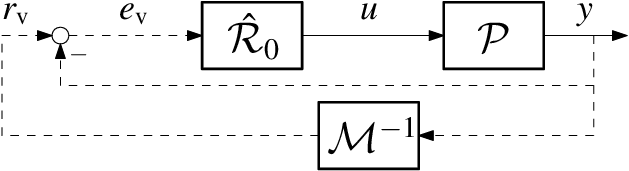}
        \caption{Generation of signals for the VRFT algorithm. Solid lines correspond to the collected data $u(t)$ and $y(t)$, being the input and output of the plant $\mP$, respectively. Dashed lines correspond instead to the virtual reference $r_\mathrm{v}(t)$ and the virtual error $e_\mathrm{v}(t)$ required by the VRFT method in order to compute the data-driven regulator $\mRdd$.}
        \label{fig:vrft}
\end{figure}
Firstly, with reference to the VRFT control scheme in Fig. \ref{fig:vrft}, define the \emph{virtual reference} $r_\mathrm{v}(t)$ as
$
r_\mathrm{v}(t) \coloneq M(s)^{-1}y(t),
$
which can be computed offline, based on the whole recorded trace of $y(t)$. Then, the \emph{virtual error} is defined as $e_\mathrm{v}(t)\coloneq r_\mathrm{v}(t)-y(t)$.
Now, exploiting \eqref{eq:R0_vrft_fr} and recalling that $\hat{u}_{0}(t) = \Rdd(s,\theta) e(t)$, the \emph{virtual input signal} $\hat{u}_{0,\mathrm{v}}(t,\theta)$ can be expressed as:
\begin{equation}\label{eq:virtual_u0}
\begin{aligned}
\hat{u}_{0,\mathrm{v}}(t,\theta) = &  \Rdd(s,\theta) e_\mathrm{v}(t) \\
= &\sum^{\overline{m}_R}_{q=0} \Pi_q(\theta)\cdot \left(\frac{s^q}{\chi(s)} e_\mathrm{v}(t)\right) 
= 
\sum^{\overline{m}_R}_{q=0} \Pi_q(\theta)\cdot \varphi_{q,\mathrm{v}}(t),
\end{aligned}
\end{equation}
where $\varphi_{q,\mathrm{v}}(t)=\frac{s^q}{\chi(s)} e_\mathrm{v}(t) \in \R^m$.
Making reference to \cite{Formentin2012}, we define matrix $\varPhi_\mathrm{v}(t)$ as
\begin{equation}
\varPhi_\mathrm{v}(t) = \begin{bmatrix} \varphi_{0,\mathrm{v}}(t)' \otimes I \cdots \varphi_{\overline{m}_R,\mathrm{v}}(t)' \otimes I \end{bmatrix}
\end{equation}
such that 
\begin{equation}\label{eq:u0_v}
\hat{u}_{0,\mathrm{v}}(t,\theta) = \varPhi_\mathrm{v}(t)\theta.
\end{equation}

Finally, the VRFT paradigm recasts Problem \ref{pb:model-ref} into the identification problem of the controller that maps the virtual error $e_\mathrm{v}(t)$ into the stored control action $u(t)$, that is minimizing
\begin{equation}
J_{\text{vr}}(\theta)=\lVert u(t)-\hat{u}_{0,\mathrm{v}}(t,\theta) \rVert_2^2,
\label{eq_mimo_vrft_obj}
\end{equation}
with $\hat{u}_{0,\mathrm{v}}(t,\theta)$ defined in \eqref{eq:u0_v}.
Exploiting \eqref{eq:u0_v}, the objective can be expressed as
\begin{equation}
\begin{aligned}
J_{\text{vr}}(\theta)&=\lVert u(t)-\varPhi_\mathrm{v}(t)\theta \rVert_2^2 \\
&= \int^{\infty}_{-\infty} \Tr\left( \left(u(\tau)-\varPhi_\mathrm{v}(\tau)\theta\right)' \left( u(\tau)-\varPhi_\mathrm{v}(\tau)\theta \right) \right) \text{d}\tau.
\label{eq_mimo_vrft_obj_2}
\end{aligned}
\end{equation}
By suitable manipulation and taking the gradient with respect to $\theta$, the minimizer $\theta_\mathrm{vr}$ of \eqref{eq_mimo_vrft_obj_2} is obtained by solving the following linear system
\begin{equation}
\left(\int^{\infty}_{-\infty} \varPhi_\mathrm{v}(\tau)' \varPhi_\mathrm{v}(\tau) \text{d}\tau \right) \theta_{\text{vr}}= \left(\int^{\infty}_{-\infty} \varPhi_\mathrm{v}(\tau)' u(\tau) \text{d}\tau \right).
\label{eq_mimo_vrft_minimizer}
\end{equation}
The data-driven controller $\hat{R}_\mathrm{vr}(s)$ is given by substituting $\theta_\mathrm{vr}$ into \eqref{eq:R0_vrft_fr}. We note that such a controller corresponds to the data-driven controller $\Rdd(s)$ defined in Section \ref{sec:mr_control}.

\section*{Acknowledgment}
The authors are indebted to Andrea Bisoffi for several enlightening discussions.

\bibliographystyle{plain}       
\bibliography{autosam}
\balance
\end{document}